\documentclass[letterpaper]{article}
\usepackage[margin=1in]{geometry}

\usepackage{ifpdf}
\ifpdf
    \usepackage[pdftex]{graphicx}
    \usepackage[update]{epstopdf}
\else
	\usepackage{graphicx}
\fi
\usepackage{wrapfig,bbm}
\usepackage{enumitem,color}
\usepackage{multirow}
\usepackage{amsthm}
\usepackage{subfigure}
\usepackage{setspace}

\newtheorem{theorem}{Theorem}

\newtheorem{lemma}{Lemma}

\newtheorem{example}{Example}
\newtheorem{prop}{\textbf{Proposition}}

\usepackage[mathscr]{euscript}
\usepackage{eufrak}
\usepackage{bbold}
\usepackage{array}
\usepackage[cmex10]{amsmath}
\usepackage{amssymb}
\usepackage{amstext}
\usepackage{cite,setspace}
\usepackage{algorithm2e}
\usepackage{url}
\usepackage[dvipsnames]{xcolor}

\newtheorem{definition}{\textbf{Definition}}

\newcommand{\Expt}{\mbox{${\mathbb E}$} }

\DeclareGraphicsExtensions{.pdf,.png,.jpg,.eps}

\usepackage{mathtools}

\newcommand{\Ic}{\mathcal{I}}

\newcommand{\Kc}{\mathcal{K}}

\newcommand{\Nc}{\mathcal{N}}

\newcommand{\Qc}{\mathcal{Q}}

\newcommand{\Sc}{\mathcal{S}}
\newcommand{\Tc}{\mathcal{T}}

\newcommand{\Nb}{\mathbb{N}}
\newcommand{\Eb}{\mathbb{E}}
\newcommand{\Fb}{\mathbb{F}}

\newcommand{\NS}{\mathrm{NS}}
\newcommand{\NB}{\mathrm{NB}}

\allowdisplaybreaks
    
\begin{document}

\title{Two-Level Private Information Retrieval}

\author{Ruida Zhou, Chao Tian, Hua Sun, and James S. Plank}
\date{}

\maketitle

\begin{abstract}
In the conventional robust $T$-colluding private information retrieval (PIR) system, the user needs to retrieve one of the possible messages while keeping the identity of the requested message private from any $T$ colluding servers. Motivated by the possible heterogeneous privacy requirements for different messages, we consider the $(N, T_1:K_1, T_2:K_2)$ two-level PIR system with a total of $K_2$ messages in the system, where $T_1\geq T_2$ and $K_1\leq K_2$. Any one of the $K_1$ messages needs to be retrieved privately against $T_1$ colluding servers, and any one of the full set of $K_2$ messages needs to be retrieved privately against $T_2$ colluding servers. We obtain a lower bound to the capacity by proposing two novel coding schemes, namely the non-uniform successive cancellation scheme and the non-uniform block cancellation scheme.  A capacity upper bound is also derived. The gap between the upper bound and the lower bounds is analyzed, and shown to vanish when $T_1=T_2$. Lastly, we show that the upper bound is in general not tight by providing a stronger bound for a special setting. 
\end{abstract}

\section{Introduction}
Capacity characterizations of the canonical private information retrieval (PIR) system and its variants have drawn considerable attention recently in the information and coding theory community, for which novel code constructions and impossibility results have been  discovered.

In the canonical PIR model, user privacy needs to be preserved during message retrieval from replicated servers, i.e., the identity of the desired message should not be revealed to any single server. Specifically, the user is required to retrieve one of the $K$ messages from $N$ servers, each of which stores a copy of $K$ messages, such that the identity of the desired message is not revealed to any single server. 
In the PIR capacity characterization problem, the goal is to identify the minimum download cost, i.e., the minimum amount of download per-bit of the desired message, the inverse of which is referred to as the capacity of PIR. The PIR capacity was characterized in \cite{sun2017capacity1} through a code construction and a matching converse bound. The code construction recursively exploits three key elements: server symmetry, message symmetry, and side information; the converse bound recursively reduces the problem scale by utilizing the privacy constraint. 

The canonical PIR problem formulation is to some extent idealized and possesses abundant symmetry and homogeneity (both in the servers and messages), which were judiciously exploited in the code construction proposed in  \cite{sun2017capacity1}. Going forward, it is imperative to enrich the canonical model to make it more heterogenous and comprehensive so that 1) practical constraints that arise naturally in diverse applications are incorporated and tackled, and 2) the extendability and limitation of the capacity results \cite{sun2017capacity1} are better understood. Along this line, the following aspects that generalize the canonical model have been studied in the literature. 

\begin{enumerate}
\item Colluding pattern: Privacy is guaranteed against each single server in the canonical model, which has been generalized to any set of $T$ colluding servers in \cite{sun2017capacity} where the optimal code construction relies on MDS coded queries built upon the same recursive procedure as in \cite{sun2017capacity1}. The $T$-colluding privacy constraint was further generalized to the fully heterogeneous model where each colluding set of servers can be an arbitrary subset of all servers \cite{tajeddine2017private, yao2020capacity}. Interestingly, while server symmetry appears to be broken, the recursively constructed MDS coded queries can still be allocated according to a linear program defined by the colluding sets, and furthermore, this elegant solution was shown to be optimal \cite{yao2020capacity}.  
\item Download per server: As the message size is allowed to approach infinity in capacity characterizations, the download size per server can be made the same through symmetrization in the canonical model \cite{tian2019capacity}. However, if other metrics are considered such as message size \cite{sun2017optimal, tian2019capacity, zhou2020capacity} or physical constraints that limit the communication link between each server and the user \cite{banawan2019asymmetry}, schemes with heterogeneous downloads per server are useful and sometimes necessary. While server symmetry is lost, the iterative construction from \cite{sun2017capacity1} can proceed with the two remaining elements in a similar manner \cite{sun2017optimal, banawan2019asymmetry, lin2018asymmetry}.
\item Message size: The $K$ messages are assumed to have equal length and allowed to approach infinity in the canonical model. The generalization to arbitrary different lengths was considered in \cite{vithana2020semantic} and the iterative construction from \cite{sun2017capacity1} was applied to truncated subsets of messages with the same length \cite{vithana2020semantic}. The minimum message sizes for capacity-achieving codes were considered in \cite{tian2019capacity, zhou2020capacity} where server symmetry and side information were utilized in the code constructions. 
\item Server storage: Each server has the same storage capability and stores all $K$ messages in the canonical model. The storage system at the servers has been generalized to MDS coded \cite{banawan2018capacity, freij2017private} or coded by a given linear code \cite{kumar2019achieving} for each message, and arbitrarily uncoded \cite{attia2018capacity} with heterogeneous capabilities \cite{banawan2020capacity, woolsey2020uncoded}. In these settings the iterative construction from \cite{sun2017capacity1} is still largely compatible with the storage structure. However, for the general model where all messages can be jointly coded, the tradeoff between the storage constraint and the download cost is far from being fully understood \cite{tian2018shannon, sun2019breaking, tian2020storage, guo2020new}.
\end{enumerate}

The main motivation of this work is a crucial aspect that has not been previously addressed - the heterogeneity of the privacy constraints on the messages. That is, in all existing works, each message is required to be equally private in the sense that any single server \cite{sun2017capacity1}, or any colluding set of $T$ servers \cite{sun2017capacity}, is completely ignorant of the desired message identity. However, the sensitivities of different types of information are commonly different in practice. To be more concrete, let us consider the following  simple example setting.

\begin{example}
There are a total of four short videos, which are replicated on six storage servers. The first two videos are political campaign videos from two opposing political parties, while the other two videos are non-political music videos. Given the sensitivity of revealing one's political view, as well as the requirement of protecting the user's privacy in a general sense, the user may wish to assure the following privacy protection when retrieving one of these videos: 
\begin{itemize}
\item Any one of the servers will not be able to infer any information regarding which message is being requested;
\item Any three of the servers jointly will not be able to infer any information regarding which one of the first two messages is being requested. 
\end{itemize}
Let us consider several scenarios to clarify the privacy protection these requirements are offering: 1) When any video is retrieved, any one of the server will not infer any information regarding which was being requested, and any four or more servers may collude to infer exactly which was being retrieved, 2) When the user retrieves one of the political campaign videos, any two or three servers may collude to infer that the retrieved video is indeed a political campaign video, but they will not be able to infer which one it is, thus protecting the user's political view; 3) When a non-political video is retrieved, any two or three servers may collude to infer exactly which video is retrieved. Therefore, the user's political view is indeed protected in a stronger manner than his preference among general contents. It should be noted that the user is not enforcing a stronger privacy protection against the fact that a political video is retrieved in general, since this fact alone does not reveal any sensitive information about the user's political preference: only the information on exactly which political video is retrieved will reveal such information. In fact, this information regarding whether a political video was requested is protected against any single server, which is offered by the lower privacy protection level.  
\end{example}


Motivated by the consideration above, we formulate the problem of multilevel private information retrieval problem. Specifically, the {\em privacy level} of a message set is defined as the maximum allowed number of colluding servers that the identity of a desired message is kept private among that message set. We focus on the two-level PIR system, where some $K_1$ messages out of the $K_2$ messages have a higher privacy level of $T_1$, i.e., any colluding set of $T_1$ servers do not learn anything about which one of the $K_1$ messages is desired, while all the $K_2$ messages together have a lower privacy level of $T_2$, any colluding set of $T_2$ servers do not learn anything about which one of the $K_2$ messages is desired. 

Characterizing the capacity of the two-level PIR system turns out to be rather challenging. A naive approach, which can be used as a baseline, is to treat the system as if it were a homogeneous $T_1$-colluding private information retrieval system. However, the crux of the two-level PIR hinges on how to leverage the less stringent privacy requirement for some messages. Towards this end, we must treat the two sets of messages with distinct privacy levels differently, i.e., message symmetry cannot be taken for granted. Without message symmetry, the iterative construction breaks since message symmetry is the key step that enables the connection between the layers, and we have to delve deeper into the code structure to adjust the parameters of the MDS coded queries in a heterogeneous manner. As a result, we discover two general schemes that can outperform the naive baseline scheme. For the converse direction, we first apply the iterative induction technique to obtain a general upper bound, and analyze the gap between the upper bound and the lower bound. We then show that this bound is strictly sub-optimal by deriving a tighter bound for a special case. This implies that the induction technique must be combined with more delicate consideration on the heterogeneous nature of the system. This observation may shed some light on other open settings, where it is not known if similar symmetric reduction based converse bounds are tight \cite{sun2017private, banawan2018multi, banawan2019asymmetry}. 

\bigskip
\noindent\textbf{Notations: }We adopt the notation $i:j \triangleq \{i, i+1, \ldots, j-1, j\}$. Denote vector $a_{\Nc} \triangleq (a_i)_{i \in \Nc}$ for any sequence $(a_1, a_2, \ldots)$ and $\Nc \subset \Nb$. We use $X \sim Y$ to indicate that the random variables $X$ and $Y$ follow an identical distribution. For any matrix $A[:,:]$, the first coordinate is for row indices and the second coordinate is for column indices.

\section{Problem Formulation} \label{sec:problem-formulation}

There are $K_2$ mutually independent messages $W_{1:K_2}=(W_1, W_2, \ldots, W_{K_2})$ in the system. Each message is uniformly distributed over $\Fb_q^{L}$, where $\Fb_q$ is a large enough finite field and $L$ is the number of $q$-ary symbols in the message (i.e., the message length). This is equivalent to
\begin{align}
&H(W_1) = H(W_2) = \cdots = H(W_{K_2}) = L, \\
&H(W_{1:K_2}) = K_2L,
\end{align}
where (and in the rest of this work) we take base-$q$ logarithm for simplicity. 
There are $N$ servers in the system, each of which stores a copy of all the $K_2$ messages. Let $k^* \in 1:K_2$ be the identity of the desired message. The process to retrieve message $W_{k^*}$, for any $k^* \in 1:K_2$, involves three steps:
\begin{enumerate}
\item[1.](Query) The user sends a randomized query $Q^{[k^*]}_{n}$ to server $n$ for each $n \in 1:N$;
\item[2.](Answer) Each server $n$, where $n \in 1:N$, returns an answer $A^{[k^*]}_n$ to the user;
\item[3.](Recovery) The user recovers the message as $\hat{W}_{k^*}$, using the queries $Q^{[k^*]}_{1:N}$ to all the servers and the answers $A^{[k^*]}_{1:N}$ from all the servers.
\end{enumerate}

Denote the set of all possible queries sent to server $n$ as $\Qc_n$. $Q^{[k^*]}_{n} \in \Qc_n$ is a random variable, whose superscript $[k^*]$ indicates that the query is for retrieving message $W_{k^*}$. The user has no knowledge of $W_{1:K_2}$, and thus the queries are independent of the messages, that is
\begin{align}
I(Q^{[k^*]}_{1:N} ; W_{1:K_2}) = 0, \quad \forall k^* \in 1:K_2.
\end{align}
Each symbol of the answer $A^{[k^*]}_n$, the answer to the query $Q^{[k^*]}_n$, is a sequence of symbols in $\Fb_q$; denote the number of symbols of $A^{[k^*]}_n$ as $\ell^{[k^*]}_n$. The answer $A^{[k^*]}_n$ is a deterministic function of the query $Q^{[k^*]}_n$ and the messages $W_{1:K_2}$, that is
\begin{align}
H(A^{[k^*]}_n | Q^{[k^*]}_n, W_{1:K_2}) = 0, \quad \forall k^* \in 1:K_2,~n \in 1:N.
\end{align}
The recovered message $\hat{W}_{k^*}$ depends on the queries $Q_{1:N}^{[k^*]}$ as well as the answers $A_{1:N}^{[k^*]}$, that is
\begin{align}
H(\hat{W}_{k^*} | A_{1:N}^{[k^*]}, Q_{1:N}^{[k^*]} ) = 0, \quad \forall k^* \in 1:K_2.
\end{align}

The message should be retrieved correctly, i.e., $W_{k^*} = \hat{W}_{k^*}$ for all $k^* \in 1:K_2$. Additionally, the system has certain privacy requirements. To measure user privacy when querying for any message in a certain set of messages, we first introduce the definition of \textit{privacy level}.

\begin{definition}[Privacy level] \label{def:privacy level}
Let the messages in the system be $W_1,W_2,\ldots,W_K$. The queries of a scheme have privacy level $T$ for a subset of messages $W_{\Sc}$, where $\Sc \subseteq 1:K$, if for any $\Tc \subseteq 1:N$ with $|\Tc| = T$, for retrieving any message in $W_\Sc$, the queries to the servers in $\Tc$ have the same joint distribution, i.e.,  
\begin{align}
Q^{[k]}_{\Tc} \sim Q^{[k']}_{\Tc}, \qquad \forall k, k' \in \Sc.
\end{align}
\end{definition}

The notion of privacy level has the following operational meaning: if $W_{\Sc}$ has privacy level $T$, then when one of the messages in $W_{\Sc}$ is retrieved, even if any $T$ of the $N$ servers collude, the identity of the requested message in $W_{\Sc}$ remains private, however these colluding servers may be able to infer that the requested message is in the set $W_\Sc$. It is straightforward to verify that the set of messages with higher privacy level automatically has lower privacy levels. In addition, when the set $\Sc$ is a singleton, if $T$ servers can infer the desired message is in $W_\Sc$, the identity of the desired message is known. Thus it is not meaningful to study the privacy level of $W_{\Sc}$ for singleton $\Sc$, though we will still allow it for notational convenience.

In this work, we consider the two-level PIR system. The system parameters in such a system are $(N, T_1: K_1, T_2 : K_2)$ with $T_1 \geq T_2\geq 1$ and $1\leq K_1 \leq K_2$. All the messages $W_{1:K_2}$ have the default weaker privacy level $T_2$, but the first $K_1$ messages $W_{1:K_1}$ have an enhanced privacy level $T_1$.  We are interested in the retrieval rate (or simply rate) which is the number of useful message symbols retrieved per unit download
\begin{align}
R \triangleq  \frac{L}{\sum_{n = 1}^N \Expt [\ell_n^{[k^*]}]}.
\end{align}
The download cost $D$ is defined as the inverse of $R$, i.e., $D \triangleq R^{-1}$. It is clear that schemes with higher achievable rates are preferred, and the supremum of the achievable rates among all possible schemes is called the capacity of the system, denoted as $C$.

\section{Main Result} \label{sec:main}

We first provide some new notation. 
Define the function $D^*_N(K, T)$ as follows
\begin{align}
D^*_N(K, T) \triangleq 1 + \frac{T}{N} + \cdots + \left( \frac{T}{N} \right)^{K-1}, \quad \forall T, K, N \in \Nb,
\end{align}
whose inverse is the capacity of the $T$-colluding PIR system with $N$ servers and $K$ messages (sometimes simply referred to as a $T$-private system). The main result of this work is summarized in the theorem below.
\begin{theorem}\label{thm:main}
The capacity $C$ of the $(N, T_1:K_1, T_2:K_2)$ two-level PIR system satisfies
\begin{align}
 \max\left( R_{\NS}, R_{\NB}\right) \leq C \leq \overline{R},
\end{align}
where
\begin{align}
\overline{R} &= \left( D_N^*\left(K_1, T_1\right) +  \frac{T_2}{N} \left( \frac{T_1}{N} \right)^{K_1-1} D_N^*(K_2-K_1, T_2)  \right)^{-1}, \label{eqn:upper} \\
R_{\NS} &= \left( D_N^*\left(K_1, T_1\right) +  \left( \frac{T_1}{N} \right)^{K_1} D_N^*(K_2-K_1, T_2)  \right)^{-1},\label{eqn:MDS-rate} \\
R_{\NB} &=\left( \max\left(  D_N^*(K_1, T_1) + \frac{T_2}{N} D_N^*(K_2-K_1, T_2),~D_N^*(K_2-K_1, T_2) + \frac{T_2}{N} D_N^*(K_1, T_1)\right) \right) ^{-1}.
\end{align}
\end{theorem}

The lower bound to the capacity in this theorem has two components: $R_{\NS}$ is obtained by the Non-uniform Successive-cancellation (NS) coding scheme given in Section \ref{sec:NS}, and $R_{\NB}$ is obtained by the Non-uniform Block-cancellation (NB) coding scheme given in Section \ref{sec:NB}. The proof for the upper bound $\overline{R}$ is given in Section \ref{sec:upper-bound}. The upper bound $\overline{R}$ in Theorem \ref{thm:main} is in general not tight. Specifically, the following proposition tightens the upper bound for the $(3, 2:2, 1:3)$ two-level PIR system, for which Theorem \ref{thm:main} gives an upper bound of $\frac{9}{17}$.
\begin{prop}\label{lem:3,2:2,1:3}
The capacity $C$ of the $(3, 2:2, 1:3)$ two-level PIR system satisfies
\begin{align}
C \leq \frac{11}{21}.
\end{align}
\end{prop}

The proof of this proposition is given in Section \ref{sec:upper-bound}, which is obtained using the computer-aided approach discussed in \cite{tian2018symmetry,tian2019open, tian2020computational}. 

To further understand these bounds in Theorem \ref{thm:main}, define 
$$\underline{D} = \overline{R}^{-1}, \quad D_{\NS} = R_{\NS}^{-1}, \quad D_{\NB} = R_{\NB}^{-1}.$$

Three observations are in order:
\begin{enumerate}
\item Theorem \ref{thm:main} gives that
\begin{align}
\underline{D} \leq \min D \leq \min\left( D_{\NS}, D_{\NB} \right). \label{eqn:download-bounds}
\end{align}
The difference between $\underline{D}$ and $D_{\NS}$ is
\begin{align}
D_{\NS} - \underline{D} = \frac{T_1 - T_2}{N} \left( \frac{T_1}{N} \right)^{K_1-1} D^*(K_2-K_1, T_2).
\end{align}
It is seen that this gap diminishes geometrically as $K_1$ grows, and also vanishes when $T_1=T_2$ as expected.

\item Any $(N, T_1:K_2, T_2:K_2)$ code, i.e., a $T_1$-private code with $N$ servers and $K_2$ messages, is valid for the $(N, T_1:K_1, T_2:K_2)$ PIR system. The optimal download cost of the former is exactly given by  $D_{\text{T-PIR}} = D_N^*(K_2, T_1)$. Comparing with this naive approach, the coding gain of the proposed NS scheme is thus
\begin{align}
D_{\text{T-PIR}} - D_{\NS} = \left(\frac{T_1}{N}\right)^{K_1} \left( D_N^*(K_2-K_1, T_1) - D_N^{*}(K_2-K_1, T_2) \right),
\end{align} 
which is always non-negative, and strictly positive if and only if $K_2-K_1 \geq 2$. Note that the strategy of using an $(N, T_1:K_2, T_2:K_2)$ code when a message in $W_{\mathcal{S}}$ is requested, and using an $(N, T_1:1, T_2:K_2)$ code for the other messages is not valid, since this would lead to privacy leakage in the latter case, i.e., leaking the information that the requested message is not in the set $\mathcal{S}$.

\item 
The relation between $R_{\NS}$ and $R_{\NB}$ is as follows. 
\begin{itemize} 
\item For the cases when
\begin{align}
D_N^*(K_1, T_1) \geq  D_{N}^{*}(K_2-K_1, T_2) \quad \text{and} \quad \frac{T_2}{N} < \left(\frac{T_1}{N}\right)^{K_1},
\end{align}
the lower bound $R_{\NB}$ is better
\begin{align}
R_{\NS} < R_{\NB} = \left(D_N^*(K_1, T_1) + \frac{T_2}{N} D_N^*(K_2-K_1, T_2) \right)^{-1};
\end{align}
\item For the cases when
\begin{align}
D_N^*(K_1, T_1) < D_{N}^{*}(K_2-K_1, T_2) \quad \text{and} \quad \frac{D_N^*(K_1, T_1)}{1 - \left(\frac{T_1}{N}\right)^{K_1}} >  \frac{D_{N}^{*}(K_2-K_1, T_2)}{1 - \frac{T_2}{N}}
\end{align}
the lower bound $R_{\NB}$ is also better
\begin{align}
R_{\NS} < R_{\NB} = \left(D_N^*(K_2-K_1, T_2) + \frac{T_2}{N} D_N^*(K_1, T_1) \right)^{-1};
\end{align}
\item For all the other cases, the lower bound $R_{\NS}$ is better, i.e., $R_{\NB} \leq R_{\NS}$.

\end{itemize}
\end{enumerate}

The upper bound and lower bounds are shown in Figure \ref{fig:bounds}. In Figure \ref{fig:K1}, the gap between the upper bound $\overline{R}$ and the rate of NS coding scheme $R_{\NS}$ deminishes geometrically as $K_1$ grows. It can be seen in Figure \ref{fig:T1}, that when $T_1$ is close to $T_2$, the NS scheme performs better, and matches the upper bound if $T_1 = T_2$; when $T_1$ is close to $N$, the NB scheme performs better, and in this case matches the upper bound if $T_1 = N$.

\begin{figure}
\hfill
\subfigure[$(10, 6:K_1, 2:K_1+4)$ two-level PIR \label{fig:K1}]{\includegraphics[scale=0.52]{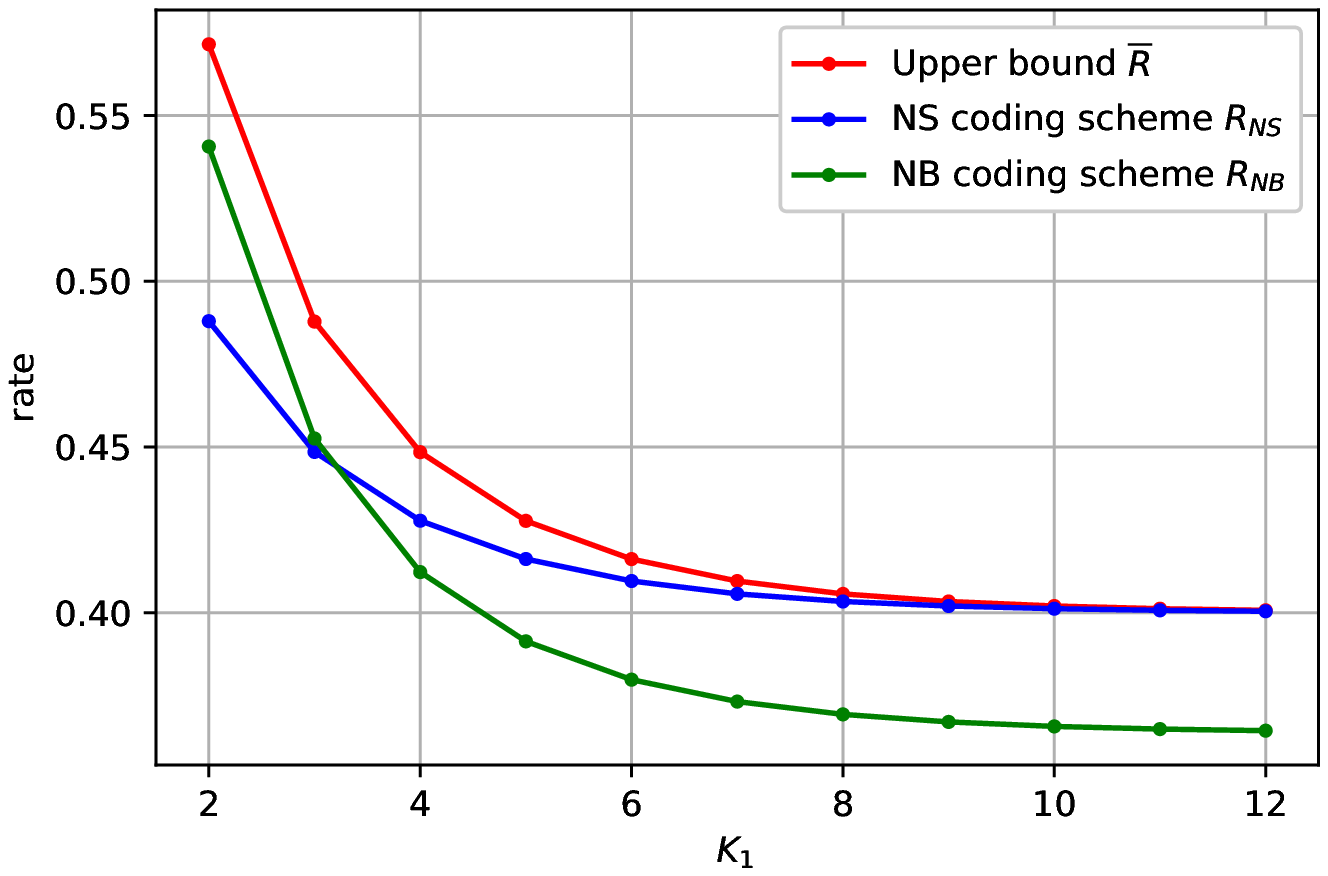}}
\hfill
\subfigure[$(10, T_1:2, 2:6)$ two-level PIR \label{fig:T1}]{\includegraphics[scale=0.52]{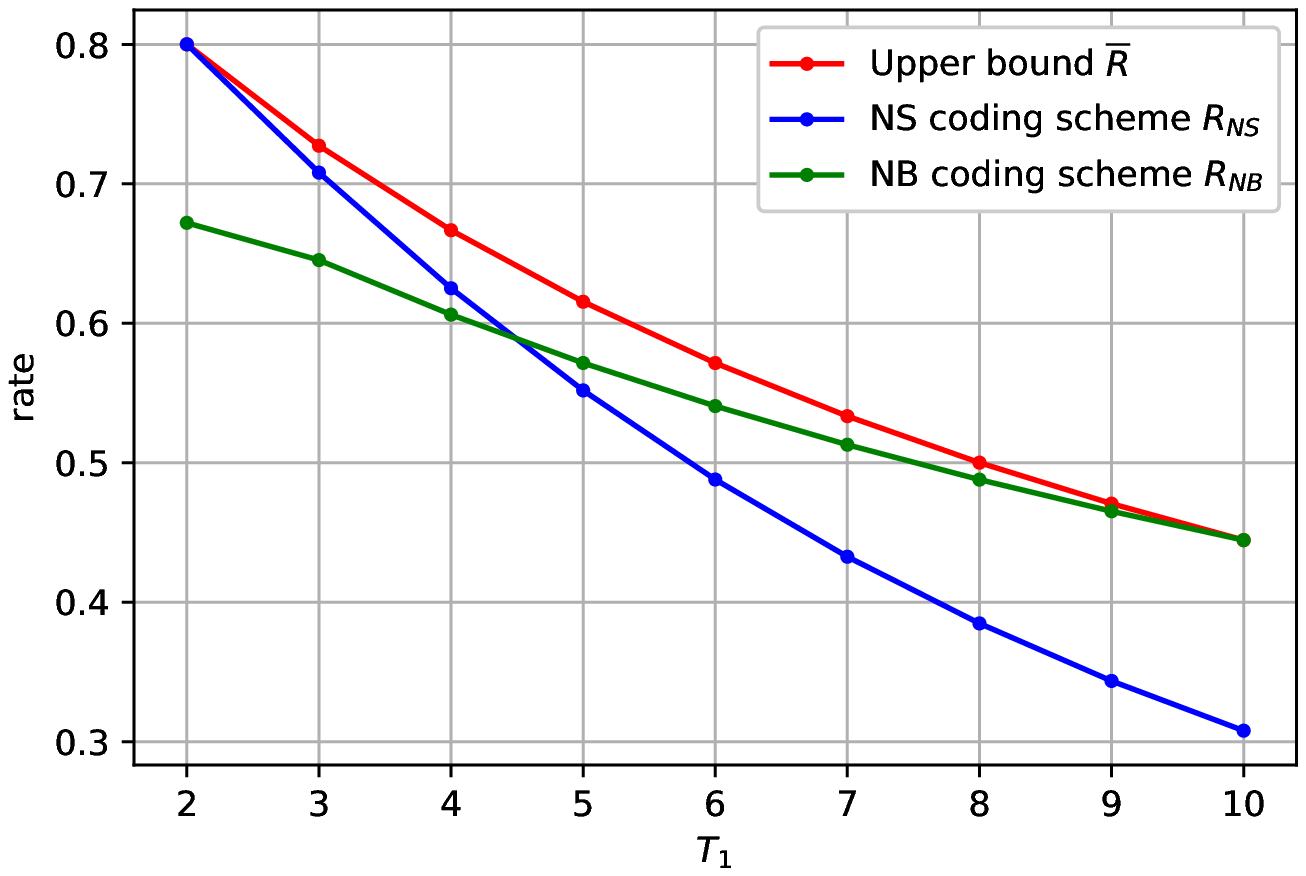}}
\hfill
\caption{Upper and lower bounds on the capacity of two-level PIR system}\label{fig:bounds}
\end{figure}

\section{A Gentle Start} 

In this section, we first provide a brief review of the $T$-colluding PIR code using two example cases, and partly based on insights obtained from these example cases, we provide two example codes to illustrate the proposed coding schemes. 

\subsection{Two $T$-colluding PIR examples}\label{sec:review}

As mentioned earlier, an $(N, T_1:K_1, T_2:K_2)$ two-level PIR system degrades to a $T_1$-private system when $K_1 = K_2$, and thus it is expected that there is a connection between the code construction for the $T$-private systems and that for the 2-level PIR systems. The capacity of the $T$-private system was identified in \cite{sun2017capacity}. We next consider two special cases of the codes proposed there, in order to provide the necessary intuition for the proposed codes. 

\begin{table}[!tb]
\centering
\caption{A  $2$-private code with $(N,K)=(4,2)$} \label{tbl:eg-SJ-2}
\resizebox{0.5\textwidth}{!}{
 \begin{tabular}{|c|c|c|c|c|c|}\hline
\mbox{Server-1}&\mbox{Server-2}&\mbox{Server-3} &\mbox{Server-4}\\\hline
$a_1, b_1$ & $a_2, b_2$ & $a_3, b_3$ & $a_4, b_4$\\ \hline \hline
$a_5+b_5$ & $a_6+b_6$ & $a_7+b_7$ & $a_8+b_8$ \\ \hline
\end{tabular}
}
\end{table}

\begin{enumerate}
\item First set $(N,T_1:K_1,T_2:K_2)=(4,2:2,1:2)$, which is essentially a $2$-private system with $N=4$ servers and $K=K_1=K_2=2$ messages. In the code given in \cite{sun2017capacity}, the message length is $8$. The messages are first precoded as $W_1^* = S_1 W_1$ and $W_2^* = S_2 W_2$, where $S_1$ and $S_2$ are random matrices drawn uniformly from the set of all $8 \times 8$ full-rank matrices over $\Fb_q$. Let $a_{1:8}$ and $b_{1:8}$ be MDS-coded symbols of messages $W^*_1$ and $W^*_2$, respectively, using appropriate coding parameters. The coding structure is given in Table \ref{tbl:eg-SJ-2}. To retrieve $W_1$, we choose $a_{1:8} = W_1^*$ and $b_{1:8}$ to be $(8, 4)$-MDS coded symbols using any $4$ symbols from $W_2^*$; the coding parameters for retrieving $W_2$ are obvious by symmetry. Since the symbols $b_{5:8}$ can be recovered from $b_{1:4}$, $W_1$ can be recovered correctly. It is not difficult to verify that the retrieval is private due to the precoding and MDS-coding steps.

\item Next let $(N,T_1:K_1,T_2:K_2)=(4,1:2,1:2)$, which is essentially the canonical PIR system with $N=4$ servers and $K=K_1=K_2=2$ messages. The coding structure is given in Table \ref{tbl:eg-SJ-1}, where $c_{1:16}$ and $d_{1:16}$ are MDS-coded symbols of two messages, respectively, with appropriate coding parameters. To retrieve $W_1$, we can use $c_{1:16} = W_1^*$ and let $d_{1:16}$ be $(16, 4)$-MDS coded of any $4$ symbols in $W_2^*$; the coding parameters for retrieving $W_2$ are obvious by symmetry. Since the symbols $d_{5:16}$ can be recovered from $d_{1:4}$, $W_1$ can be recovered correctly and privately. 

\end{enumerate}
\begin{table}[!tb]
\centering
\caption{A  $1$-private code with $(N,K)=(4,2)$} \label{tbl:eg-SJ-1}
\resizebox{0.5\textwidth}{!}{
 \begin{tabular}{|c|c|c|c|c|c|}\hline
\mbox{Server-1}&\mbox{Server-2}&\mbox{Server-3} &\mbox{Server-4}\\\hline
$c_1, d_1$ & $c_2, d_2$ & $c_3, d_3$ & $c_4, d_4$ \\ \hline \hline
$c_5+d_5$ & $c_6+d_6$ & $c_7+d_7$ & $c_8+d_8$ \\
$c_9+d_9$ & $c_{10}+d_{10}$ & $c_{11}+d_{11}$ & $c_{12}+d_{12}$ \\
$c_{13}+d_{13}$ & $c_{14}+d_{14}$ & $c_{15}+d_{15}$ & $c_{16}+d_{16}$ \\ \hline
\end{tabular}
}
\end{table}

Comparing these two cases, a few observations are now in order:
\begin{enumerate}
\item The codes in Table \ref{tbl:eg-SJ-2}  and Table \ref{tbl:eg-SJ-1} have two layers: the first layer has single symbols, i.e., $a$, $b$, $c$,  or $d$, and the second layer has summations of two symbols, i.e.,  $a+b$ or $c+d$.
\item Although the $2$-private code meets the privacy requirement of the $1$-private system, the coding structure in Table \ref{tbl:eg-SJ-1} is more efficient. Particularly, the ratio between the first layer transmissions and second layer transmissions changes from $8:4$ to $8:12$. Placing more symbols in the second layer is preferable, because  one desired symbol essentially takes two symbol transmissions in the first layer, yet it takes only one in the second layer. 
\item The improved transmission ratio between the two layers is a consequence of the chosen MDS coding parameters for the non-requested message (i.e., the interference): in Table \ref{tbl:eg-SJ-2}, it is $(8,4)$ while in Table \ref{tbl:eg-SJ-1} it is $(16,4)$. These parameters, which are chosen to satisfy the decoding and privacy requirements, determine the number of symbols in different layers. 
\end{enumerate}
These observations suggest that in a two-level PIR system, we will need to adjust the MDS coding parameters for different messages according to their privacy levels, but maintain the code structure consistent between the two cases when retrieving two types of messages. This is a considerable generalization of the $T$-private setting, since  in the $T$-private setting the MDS coding parameters can be chosen uniformly for all the messages, except the requested message, while in our setting, the privacy levels create heterogeneity among the messages.  

In the code construction given in \cite{sun2017capacity}, the following lemma plays an instrumental role in formally showing the privacy condition to hold, which we shall also utilize in this work. 
\begin{lemma}[Statistical effect of full rank matrices \cite{sun2017capacity}] \label{lem:equivlent}
Let $S_1, S_2, \ldots, S_K \in \Fb_q^{\alpha \times \alpha}$ be $K$ random matrices, drawn independently and uniformly from all $\alpha \times \alpha$ full-rank matrices over $\Fb_q$. Let $G_1, G_2, \ldots, G_K \in \Fb_q^{\beta \times \beta}$ be $K$ invertible square matrices of dimension $\beta \times \beta$ over $\Fb_q$. Let $\Ic_1,\Ic_2,\ldots ,\Ic_K \in \Nb^{\beta \times 1}$ be $K$ index vectors, each containing $\beta$ distinct indices from $[1: \alpha]$. Then
\begin{align}
(G_1 S_1[\Ic_1, :], G_2 S_2[\Ic_2, :], \ldots, G_K S_K[\Ic_K, :]) \sim (S_1[1:\beta, :], S_2[1:\beta, :], \ldots, S_K[1:\beta, :]),
\end{align}
where the notation $S[\Ic,:]$ is used to indicate the submatrix of $S$ by taking its rows in $\Ic$.
\end{lemma}

\subsection{An example of the NS scheme} \label{sec:eg-NS}
We next provide an example to illustrate the proposed NS coding scheme.  In this example, the two-level PIR system is specified by the parameters $(N, T_1: K_1, T_2:K_2)=(4,2:2,1:4)$, i.e., there are $4$ servers and $4$ messages $W_{1:4}$, and messages $W_{1:2}$ have privacy level $T_1=2$, while all messages $W_{1:4}$ have privacy level $T_2=1$. The length of each message is $L=64$ here. 

\begin{table}[h]
\caption{NS Scheme in $(N, T_1: K_1, T_2:K_2) = (4, 2:2, 1:4)$  for retrieving $W_1$} \label{tbl:eg-NS-a}
\resizebox{1.01\textwidth}{!}{
 \begin{tabular}{|c||c|c|c|c|c|}\hline
\mbox{Coding group}&\mbox{Server-1}&\mbox{Server-2}&\mbox{Server-3} & \mbox{Server-4}\\
\hline
$a$: (64, 64) & $a_1, a_2, a_3$ & $a_4, a_5, a_6$ & $a_7, a_8, a_9$ & $a_{10}, a_{11}, a_{12}$ \\ 
\hline
\textcolor{Mulberry}{$b$}: (24, 12)&\textcolor{Mulberry}{$b_1, b_2, b_3$} & \textcolor{Mulberry}{$b_4, b_5, b_6$} & \textcolor{Mulberry}{$b_7, b_8, b_9$} & \textcolor{Mulberry}{$b_{10}, b_{11}, b_{12}$} \\ 
\textcolor{blue}{$c$}: (8, 4) & \textcolor{blue}{$c_1$} & \textcolor{blue}{$c_2$}& \textcolor{blue}{$c_3$} & \textcolor{blue}{$c_4$}\\
\textcolor{green}{$d$}: (8, 4) & \textcolor{green}{$d_1$} & \textcolor{green}{$d_2$} & \textcolor{green}{$d_3$} & \textcolor{green}{$d_4$}\\
\hline\hline
& $a_{13} + \textcolor{Mulberry}{b_{13}}$ & $a_{14} + \textcolor{Mulberry}{b_{14}}$ & $a_{15} + \textcolor{Mulberry}{b_{15}}$ & $a_{16} + \textcolor{Mulberry}{b_{16}}$ \\
& $a_{17} + \textcolor{Mulberry}{b_{17}}$ &  $a_{18} + \textcolor{Mulberry}{b_{18}}$ &  $a_{19} + \textcolor{Mulberry}{b_{19}}$ & $a_{20} + \textcolor{Mulberry}{b_{20}}$ \\
& $a_{21} + \textcolor{Mulberry}{b_{21}}$ & $a_{22} + \textcolor{Mulberry}{b_{22}}$  & $a_{23} + \textcolor{Mulberry}{b_{23}}$  & $a_{24} + \textcolor{Mulberry}{b_{24}}$ \\ 
& $a_{25}+ \textcolor{blue}{c_{5}}$ & $a_{26}+ \textcolor{blue}{c_{6}}$ & $a_{27}+ \textcolor{blue}{c_{7}}$ & $a_{28}+ \textcolor{blue}{c_{8}}$ \\
& $a_{29}+ \textcolor{green}{d_{5}}$ & $a_{30}+ \textcolor{green}{d_{6}}$ & $a_{31}+ \textcolor{green}{d_{7}}$ & $a_{32}+ \textcolor{green}{d_{8}}$ \\
\hline 
\textcolor{red}{$b,c$}: (8,4) & $\textcolor{red}{b_{25}+c_{9}}$ & $\textcolor{red}{b_{26}+c_{10}}$ & $\textcolor{red}{b_{27}+c_{11}}$ & $\textcolor{red}{b_{28}+c_{12}}$ \\
\textcolor{cyan}{$b,d$}: (8,4) & $\textcolor{cyan}{b_{29}+d_{9}}$ & $\textcolor{cyan}{b_{30}+d_{10}}$ & $\textcolor{cyan}{b_{31}+d_{11}}$ & $\textcolor{cyan}{b_{32}+d_{12}}$ \\
\textcolor{Dandelion}{$c,d$}: (24, 12)& $\textcolor{Dandelion}{c_{13}+d_{13}}$ & $\textcolor{Dandelion}{c_{16}+d_{16}}$ & $\textcolor{Dandelion}{c_{19}+d_{19}}$ & $\textcolor{Dandelion}{c_{22}+d_{22}}$ \\
 & $\textcolor{Dandelion}{c_{14}+d_{14}}$ & $\textcolor{Dandelion}{c_{17}+d_{17}}$ & $\textcolor{Dandelion}{c_{20}+d_{20}}$ & $\textcolor{Dandelion}{c_{23}+d_{23}}$ \\
& $\textcolor{Dandelion}{c_{15}+d_{15}}$ & $\textcolor{Dandelion}{c_{18}+d_{18}}$ & $\textcolor{Dandelion}{c_{21}+d_{21}}$ & $\textcolor{Dandelion}{c_{24}+d_{24}}$ \\
\hline\hline
& $a_{33}+ \textcolor{red}{b_{33}+c_{25}}$ & $a_{34}+ \textcolor{red}{b_{34}+c_{26}}$ & $a_{35}+\textcolor{red}{b_{35}+c_{27}}$ & $a_{36}+ \textcolor{red}{b_{36}+c_{28}}$ \\
& $a_{37}+ \textcolor{cyan}{b_{37}+d_{25}}$ & $a_{38}+ \textcolor{cyan}{b_{38}+d_{26}}$ & $a_{39}+ \textcolor{cyan}{b_{39}+d_{27}}$ & $a_{40}+ \textcolor{cyan}{b_{40}+d_{28}}$ \\
& $a_{41}+\textcolor{Dandelion}{c_{29}+d_{29}}$ & $a_{42}+ \textcolor{Dandelion}{c_{30}+d_{30}}$ & $a_{43}+ \textcolor{Dandelion}{c_{31}+d_{31}}$ & $a_{44}+ \textcolor{Dandelion}{c_{32}+d_{32}}$ \\
& $a_{45}+ \textcolor{Dandelion}{c_{33}+d_{33}}$ & $a_{46}+ \textcolor{Dandelion}{c_{34}+d_{34}}$ & $a_{47}+ \textcolor{Dandelion}{c_{35}+d_{35}}$ & $a_{48}+ \textcolor{Dandelion}{c_{36}+d_{36}}$ \\
& $a_{49}+ \textcolor{Dandelion}{c_{37}+d_{38}}$ & $a_{50}+ \textcolor{Dandelion}{c_{38}+d_{38}}$ & $a_{51}+ \textcolor{Dandelion}{c_{39}+d_{39}}$ & $a_{52}+ \textcolor{Dandelion}{c_{40}+d_{40}}$ \\ 
\hline
\textcolor{magenta}{$b,c,d$}: (24, 12)& $\textcolor{magenta}{b_{41}+c_{41}+d_{41}}$ & $\textcolor{magenta}{b_{42}+c_{42}+d_{42}}$ & $\textcolor{magenta}{b_{43}+c_{43}+d_{43}}$ & $\textcolor{magenta}{b_{44}+c_{44}+d_{44}}$ \\
 & $\textcolor{magenta}{b_{45}+c_{45}+d_{45}}$ & $\textcolor{magenta}{b_{46}+c_{46}+d_{46}}$ & $\textcolor{magenta}{b_{47}+c_{47}+d_{47}}$ & $\textcolor{magenta}{b_{48}+c_{48}+d_{48}}$ \\
& $\textcolor{magenta}{b_{49}+c_{49}+d_{49}}$ & $\textcolor{magenta}{b_{50}+c_{50}+d_{50}}$ & $\textcolor{magenta}{b_{51}+c_{51}+d_{51}}$ & $\textcolor{magenta}{b_{52}+c_{52}+d_{52}}$ \\
\hline\hline
& $a_{53}+ \textcolor{magenta}{b_{53}+c_{53}+d_{53}}$ & $a_{54}+ \textcolor{magenta}{b_{54}+c_{54}+d_{54}}$ & $a_{55}+ \textcolor{magenta}{b_{55}+c_{55}+d_{55}}$ & $a_{56}+ \textcolor{magenta}{b_{56}+c_{56}+d_{56}}$ \\
& $a_{57}+ \textcolor{magenta}{b_{57}+c_{57}+d_{57}}$ & $a_{58}+ \textcolor{magenta}{b_{58}+c_{58}+d_{58}}$ & $a_{59}+ \textcolor{magenta}{b_{59}+c_{59}+d_{59}}$ & $a_{60}+ \textcolor{magenta}{b_{60}+c_{60}+d_{60}}$\\
& $a_{61}+ \textcolor{magenta}{b_{61}+c_{61}+d_{61}}$ & $a_{62}+ \textcolor{magenta}{b_{62}+c_{62}+d_{62}}$ & $a_{63}+ \textcolor{magenta}{b_{63}+c_{63}+d_{63}}$ & $a_{64}+ \textcolor{magenta}{b_{64}+c_{64}+d_{64}}$ \\
\hline
\end{tabular}}
\end{table}

\begin{table}[h]
\caption{NS Scheme in $(N, T_1: K_1, T_2:K_2) = (4, 2:2, 1:4)$ for retrieving $W_4$} \label{tbl:eg-NS-d}
\resizebox{1.01\textwidth}{!}{
 \begin{tabular}{|c||c|c|c|c|c|}\hline
Coding group&\mbox{Server-1}&\mbox{Server-2}&\mbox{Server-3} & \mbox{Server-4}\\\hline
$d$: (64, 64)& $d_1$ & $d_2$ & $d_3$ & $d_4$ \\
\hline
\textcolor{Mulberry}{$a$}: (16, 4)& $\textcolor{Mulberry}{a_1, a_2, a_3}$ & $\textcolor{Mulberry}{a_4, a_5, a_6}$ & $\textcolor{Mulberry}{a_7, a_8, a_9}$ & $\textcolor{Mulberry}{a_{10}, a_{11}, a_{12}}$ \\ 
\textcolor{blue}{$b$}: (16, 4)& $\textcolor{blue}{b_1, b_2, b_3}$ & $\textcolor{blue}{b_4, b_5, b_6}$ & $\textcolor{blue}{b_7, b_8, b_9}$ & $\textcolor{blue}{b_{10}, b_{11}, b_{12}}$ \\ 
\textcolor{green}{$c$}: (16, 4) & $\textcolor{green}{c_1}$ & $\textcolor{green}{c_2}$ & $\textcolor{green}{c_3}$ & $\textcolor{green}{c_4}$ \\
\hline \hline
& $\textcolor{Mulberry}{a_{29}}+d_{5}$ & $\textcolor{Mulberry}{a_{30}}+d_{6}$ & $\textcolor{Mulberry}{a_{31}}+d_{7}$ & $\textcolor{Mulberry}{a_{32}}+d_{8}$ \\ 
& $\textcolor{blue}{b_{29}}+d_{9}$ & $\textcolor{blue}{b_{30}}+d_{10}$ & $\textcolor{blue}{b_{31}}+d_{11}$ & $\textcolor{blue}{b_{32}}+d_{12}$ \\
& $\textcolor{green}{c_{13}}+d_{13}$ & $\textcolor{green}{c_{16}}+d_{16}$ & $\textcolor{green}{c_{19}}+d_{19}$ & $\textcolor{green}{c_{22}}+d_{22}$ \\
& $\textcolor{green}{c_{14}}+d_{14}$ & $\textcolor{green}{c_{17}}+d_{17}$ & $\textcolor{green}{c_{20}}+d_{20}$ & $\textcolor{green}{c_{23}}+d_{23}$ \\
& $\textcolor{green}{c_{15}}+d_{15}$ & $\textcolor{green}{c_{18}}+d_{18}$ & $\textcolor{green}{c_{21}}+d_{21}$ & $\textcolor{green}{c_{24}}+d_{24}$ \\
\hline
$\textcolor{red}{a,b}: (16, 4)$ & $\textcolor{red}{a_{13} +b_{13}}$ & $\textcolor{red}{a_{14} +b_{14}}$ & $\textcolor{red}{a_{15} +b_{15}}$ & $\textcolor{red}{a_{16} +b_{16}}$ \\
& $\textcolor{red}{a_{17} +b_{17}}$ & $\textcolor{red}{a_{18} +b_{18}}$ & $\textcolor{red}{a_{19} +b_{19}}$ & $\textcolor{red}{a_{20} +b_{20}}$ \\
& $\textcolor{red}{a_{21} +b_{21}}$ & $\textcolor{red}{a_{22} +b_{22}}$  & $\textcolor{red}{a_{23} +b_{23}}$  & $\textcolor{red}{a_{24} +b_{24}}$ \\ 
\textcolor{cyan}{$a,c$}: (16,4) & $\textcolor{cyan}{a_{25} +c_{5}}$ & $\textcolor{cyan}{a_{26} +c_{6}}$ & $\textcolor{cyan}{a_{27} +c_{7}}$ & $\textcolor{cyan}{a_{28}+c_{8}}$ \\
\textcolor{Dandelion}{$b,c$}: (16,4) & $\textcolor{Dandelion}{b_{25}+c_{9}}$ & $\textcolor{Dandelion}{b_{26}+c_{10}}$ & $\textcolor{Dandelion}{b_{27}+c_{11}}$ & $\textcolor{Dandelion}{b_{28}+c_{12}}$ \\
\hline\hline
& $\textcolor{red}{a_{37}+b_{37}}+d_{25}$ & $\textcolor{red}{a_{38}+b_{38}}+d_{26}$ & $\textcolor{red}{a_{39}+b_{39}}+d_{27}$ & $\textcolor{red}{a_{40}+b_{40}}+d_{28}$ \\
& $\textcolor{cyan}{a_{41}+c_{29}}+d_{29}$ & $\textcolor{cyan}{a_{42}+c_{30}}+d_{30}$ & $\textcolor{cyan}{a_{43}+c_{31}}+d_{31}$ & $\textcolor{cyan}{a_{44}+c_{32}}+d_{32}$ \\
& $\textcolor{cyan}{a_{45}+c_{33}}+d_{33}$ & $\textcolor{cyan}{a_{46}+c_{34}}+d_{34}$ & $\textcolor{cyan}{a_{47}+c_{35}}+d_{35}$ & $\textcolor{cyan}{a_{48}+c_{36}}+d_{36}$ \\
& $\textcolor{cyan}{a_{49}+c_{37}}+d_{38}$ & $\textcolor{cyan}{a_{50}+c_{38}}+d_{38}$ & $\textcolor{cyan}{a_{51}+c_{39}}+d_{39}$ & $\textcolor{cyan}{a_{52}+c_{40}}+d_{40}$ \\
& $\textcolor{Dandelion}{b_{41}+c_{41}}+d_{41}$ & $\textcolor{Dandelion}{b_{42}+c_{42}}+d_{42} $ & $\textcolor{Dandelion}{b_{43}+c_{43}}+d_{43}$ & $\textcolor{Dandelion}{b_{44}+c_{44}}+d_{44}$ \\
& $\textcolor{Dandelion}{b_{45}+c_{45}}+d_{45}$ & $\textcolor{Dandelion}{b_{46}+c_{46}}+d_{46} $& $\textcolor{Dandelion}{b_{47}+c_{47}}+d_{47}$ & $\textcolor{Dandelion}{b_{48}+c_{48}}+d_{48}$ \\
& $\textcolor{Dandelion}{b_{49}+c_{49}}+d_{49}$ & $\textcolor{Dandelion}{b_{50}+c_{50}}+d_{50}$ & $\textcolor{Dandelion}{b_{51}+c_{51}}+d_{51}$ & $\textcolor{Dandelion}{b_{52}+c_{52}}+d_{52}$ \\
\hline
\textcolor{magenta}{$a,b,c$}: (16,4) & $\textcolor{magenta}{a_{33}+b_{33}+c_{25}}$ & $\textcolor{magenta}{a_{34}+b_{34}+c_{26}}$ & $\textcolor{magenta}{a_{35}+b_{35}+c_{27}}$ & $\textcolor{magenta}{a_{36}+b_{36}+c_{28}}$ \\
\hline \hline
& $\textcolor{magenta}{a_{53}+b_{53}+c_{53}}+d_{53}$ & $\textcolor{magenta}{a_{54}+b_{54}+c_{54}}+d_{54}$ & $\textcolor{magenta}{a_{55}+b_{55}+c_{55}}+d_{55}$ & $\textcolor{magenta}{a_{56}+b_{56}+c_{56}}+d_{56}$ \\
& $\textcolor{magenta}{a_{57}+b_{57}+c_{57}}+d_{57}$ & $\textcolor{magenta}{a_{58}+b_{58}+c_{58}}+d_{58}$ & $\textcolor{magenta}{a_{59}+b_{59}+c_{59}}+d_{59}$ & $\textcolor{magenta}{a_{60}+b_{60}+c_{60}}+d_{60}$ \\
& $\textcolor{magenta}{a_{61}+b_{61}+c_{61}}+d_{61}$ & $\textcolor{magenta}{a_{62}+b_{62}+c_{62}}+d_{62}$ & $\textcolor{magenta}{a_{63}+b_{63}+c_{63}}+d_{63}$ & $\textcolor{magenta}{a_{64}+b_{64}+c_{64}}+d_{64}$ \\
\hline
\end{tabular}}
\end{table}

\vspace{0.2cm}
\noindent\textbf{Encoding:} To retrieve a message, the answers are formed in three steps, and the queries are simply the encoding matrix for these answers. Assume for each $(n,k)$ pair where $n\geq k$, an MDS code in $\Fb_q$ is given and fixed, and we refer to it as the $(n,k)$ MDS code. The coding structure is illustrated in Table \ref{tbl:eg-NS-a} and Table \ref{tbl:eg-NS-d}, for the retrieval of $W_1$ and $W_4$, respectively.
The coding steps can be understood as follows:
\begin{enumerate}
\item \textit{Precoding:} Let $S_{1}, S_2, S_3$,  and $S_4$ be four random matrices, which are independently and uniformly drawn from the set of all $64 \times 64$ full rank matrices over $\Fb_q$; these matrices are known only to the user. The precoded messages $W^*_{1:4}$ are  
\begin{align}
W^*_1 = S_1 W_1; \quad W^*_2 = S_2 W_2; \quad W^*_3 = S_3 W_3; \quad W^*_4 = S_4 W_4.
\end{align}
\item \textit{Group-wise MDS coding:} The precoded messages are partitioned into non-overlapping segments, and each segment is MDS-coded under certain $(n,k)$ parameters, the result of which is referred to as a coding group. These MDS-coded symbols for the four messages are denoted as $a_{1:64}, b_{1:64}, c_{1:64}, d_{1:64}$, respectively. In the tables, these coding groups are distinguished using different colors, with the corresponding MDS parameters given in the first column. For example, the red coding groups in Table \ref{tbl:eg-NS-a} for both \textcolor{red}{$b_{25:28,33:36}$} and \textcolor{red}{$c_{9:12,25:28}$} are obtained by encoding $4$ pre-coded symbols in $W^*_2$ and $W^*_4$, respectively. In each coding group, the coded symbols are ordered and sequentially placed in the tables, indicated by their subscripts. 
\item \textit{Forming pre-coded message sums:} The summations of the MDS-coded messages are formed accordingly, which can be seen clearly from Table \ref{tbl:eg-NS-a} and Table \ref{tbl:eg-NS-d}.
\end{enumerate}

\vspace{0.2cm}
\noindent\textbf{Decoding and correctness:} The coding structure is layered, where in each layer the number of summands in each downloaded symbol is the same. From top to bottom, the number of summands increases from $1$ to $4$. The symbols of interference messages in each coding group are placed in two adjacent layers, where the signals (i.e., the summation of the symbols of interference messages) in the top layer can decode the interference signals in lower layer due to the common linear MDS code.

In Table \ref{tbl:eg-NS-a}, for each coding group, the total number of interference signals placed in two adjacent layers and the top layer follow the ratio $(2:1)=(8:4)=(24:12)$. For example, $8$ interference signals in the red coding group are placed in the second and third layers, where $4$ downloaded symbols $\textcolor{red}{b_{25:28} + c_{9:12}}$ in the second layer can decode $\textcolor{red}{b_{33:36} + c_{25:28}}$ in the third layer, because $\textcolor{red}{b,c}$ are encoded by the same linear $(8, 4)$ MDS code. Consequently, $a_{33:36}$ can be recovered. It can be verified that $a_{1:64}$ can all be recovered either directly or in this fashion.
By symmetry, $W_2$ can be retrieved similarly. 

In Table \ref{tbl:eg-NS-d}, for each coding group, the numbers of interference signals of each coding group placed in two adjacent layers and the top layer have the ratio at most $4:1$. For example, $16$ interference signals in red coding groups are placed in the second and third layers, where any $4$ of the $12$ downloaded symbols $\textcolor{red}{a_{13:24} + b_{13:24}}$ in the second layer can decode $\textcolor{red}{a_{37:40} + b_{37:40}}$ in the third layer because $\textcolor{red}{a,b}$ are encoded by the same linear $(16, 4)$ MDS code. Consequently, $d_{25:28}$ can be recovered. It can be verified that $d_{1:64}$ can all be recovered either directly, or in this fashion. By symmetry, $W_3$ can be retrieved similarly. 

\vspace{0.2cm}
\noindent\textbf{Privacy:} The coding pattern, i.e., the manner of forming pre-coded message sums, is the same for the retrieval of any message in $W_{1:4}$.
Since it is a linear code, the coded symbols can be generated by the corresponding coding matrices. From Table \ref{tbl:eg-NS-a}, it is seen that the coding matrix of the coded symbols of any message from any two servers has full row-rank. For examples, the coded symbols $a$'s in server-1 and server-2 can be generated by a full row rank coding matrix using the message $W_1$, due to the pre-coding and the group-wise MDS coding. By applying Lemma \ref{lem:equivlent}, the messages $W_{1:2}$ thus have privacy level $2$. The $1$-privacy for all the messages can be seen in a similar manner. 

\vspace{0.2cm}
\noindent\textbf{Performance:} The total number of downloaded symbols is $116$ and the message length is $64$. Thus the rate is $R_{\NS} = \frac{64}{116} = \frac{16}{29}$. The scheme for $2$-private systems has rate $\frac{8}{15} < R_{\NS}$.
\vspace{0.2cm}

\noindent\textit{Remark:} The construction resembles the scheme in \cite{sun2017capacity} (also discussed in Section \ref{sec:review}), but it allows non-uniform coding structure to leverage the requirements of two levels of privacy. 
Due to the homogeneity of the privacy requirements for all the messages in $T$-private systems, the MDS coding parameters for each coding group are chosen to be  $(N, T)$. In the proposed scheme for the $(N, T_1:K_1, T_2:K_2)$ system, there is symmetry among servers, and also symmetries among $W_{1:K_1}$ and among $W_{K_1+1:K_2}$ but not across all the messages. Thus when retrieving message $W_{k^*}$ with $k^* \in 1:K_1$, the ratio of the MDS parameters $(n, k)$ in each coding group of the undesired messages need to be chosen as $(N,T_1)$, while as for message $W_{k^*}$ with $k^* \in K_1+1:K_2$, the MDS coding parameters in each coding group would be $(N,T_2)$.
However, since $ N/T_1 < N/T_2$, with the same retrieval pattern, there exists certain slack in the placement pattern when retrieving $W_{k^*}$ with $k^* \in K_1+1:K_2$. For example, the red coding group in Table \ref{tbl:eg-NS-d} only needs $4$ symbols in layer-$2$ to decode the remaining symbols in both layer-$2$  and layer-$3$, yet $12$ symbols are retrieved and available directly in layer-$2$.

\subsection{An example of the NB scheme} \label{sec:eg-NB}
We provide an example to illustrate the proposed NB coding scheme for the same two-level PIR system specified by paramters $(N, T_1:K_1, T_2:K_2) = (4, 2:2, 1:4)$. The length of each message is again $L = 64$.

\begin{table}[h]
\centering
\caption{NB Scheme in $(N, T_1: K_1, T_2:K_2) = (4, 2:2, 1:4)$ for retrieving $W_1$} \label{tbl:eg-NB-a}
\resizebox{0.98\textwidth}{!}{
\begin{tabular}{|c||c|c|c|c|c|}\hline
Coding group&\mbox{Server-1}&\mbox{Server-2}&\mbox{Server-3} & \mbox{Server-4}\\
\hline
$a: (64, 64)$ & $a_1, \textcolor{red}{b_1}$ & $a_3, \textcolor{red}{b_3}$ & $a_5, \textcolor{red}{b_5}$ & $a_7, \textcolor{red}{b_7}$ \\
$\textcolor{red}{b}: (64, 32)$ & $a_2, \textcolor{red}{b_2}$ & $a_4, \textcolor{red}{b_4}$ & $a_6, \textcolor{red}{b_6}$ & $a_8, \textcolor{red}{b_8}$ \\
\hline\hline
& $a_{9} + \textcolor{red}{b_{9}}$ & $a_{11} + \textcolor{red}{b_{11}}$ & $a_{13} + \textcolor{red}{b_{13}}$ & $a_{15} + \textcolor{red}{b_{15}}$ \\
& $a_{10} + \textcolor{red}{b_{10}}$ & $a_{12} + \textcolor{red}{b_{12}}$ & $a_{14} + \textcolor{red}{b_{14}}$ & $a_{16} + \textcolor{red}{b_{16}}$ \\
\hline\hline\hline
$\textcolor{Mulberry}{c}: (16, 4)$; $\textcolor{blue}{d}: (16, 4)$  & $\textcolor{Mulberry}{c_1}, \textcolor{blue}{d_1}$ & $\textcolor{Mulberry}{c_2}, \textcolor{blue}{d_2}$ & $\textcolor{Mulberry}{c_3}, \textcolor{blue}{d_3}$ & $\textcolor{Mulberry}{c_4}, \textcolor{blue}{d_4}$ \\
\hline\hline
$\textcolor{green}{c+d}: (48, 12)$ & $\textcolor{green}{c_5+d_5}$ & $\textcolor{green}{c_8+d_8}$ & $\textcolor{green}{c_{11}+d_{11}}$ & $\textcolor{green}{c_{14}+d_{14}}$ \\
& $\textcolor{green}{c_6+d_6}$ & $\textcolor{green}{c_9+d_9}$ & $\textcolor{green}{c_{12}+d_{12}}$ & $\textcolor{green}{c_{15}+d_{15}}$ \\
& $\textcolor{green}{c_7+d_7}$ & $\textcolor{green}{c_{10}+d_{10}}$ & $\textcolor{green}{c_{13}+d_{13}}$ & $\textcolor{green}{c_{16}+d_{16}}$ \\
\hline\hline\hline
& $a_{17}+\textcolor{Mulberry}{c_{17}}$, $\textcolor{red}{b_{17}}+\textcolor{blue}{d_{17}}$ & $a_{18}+\textcolor{Mulberry}{c_{18}}$, $\textcolor{red}{b_{18}}+\textcolor{blue}{d_{18}}$ & $a_{19}+\textcolor{Mulberry}{c_{19}}$, $\textcolor{red}{b_{19}}+\textcolor{blue}{d_{19}}$ & $a_{20}+\textcolor{Mulberry}{c_{20}}$, $\textcolor{red}{b_{20}}+\textcolor{blue}{d_{20}}$ \\
 & $a_{21}+\textcolor{Mulberry}{c_{21}}$, $\textcolor{red}{b_{21}}+\textcolor{blue}{d_{21}}$ & $a_{22}+\textcolor{Mulberry}{c_{22}}$, $\textcolor{red}{b_{22}}+\textcolor{blue}{d_{22}}$ & $a_{23}+\textcolor{Mulberry}{c_{23}}$, $\textcolor{red}{b_{23}}+\textcolor{blue}{d_{23}}$ & $a_{24}+\textcolor{Mulberry}{c_{24}}$, $\textcolor{red}{b_{24}}+\textcolor{blue}{d_{24}}$ \\
& $a_{25}+\textcolor{Mulberry}{c_{25}}$, $\textcolor{red}{b_{25}}+\textcolor{blue}{d_{25}}$ & $a_{26}+\textcolor{Mulberry}{c_{26}}$, $\textcolor{red}{b_{26}}+\textcolor{blue}{d_{26}}$ & $a_{27}+\textcolor{Mulberry}{c_{27}}$, $\textcolor{red}{b_{27}}+\textcolor{blue}{d_{27}}$ & $a_{28}+\textcolor{Mulberry}{c_{28}}$, $\textcolor{red}{b_{28}}+\textcolor{blue}{d_{28}}$ \\
& $a_{29}+\textcolor{green}{c_{29}+d_{29}}$ & $a_{32}+\textcolor{green}{c_{32}+d_{32}}$ & $a_{35}+\textcolor{green}{c_{35}+d_{35}}$ & $a_{38}+\textcolor{green}{c_{38}+d_{38}}$ \\
& $a_{30}+\textcolor{green}{c_{30}+d_{30}}$ & $a_{33}+\textcolor{green}{c_{33}+d_{33}}$ & $a_{36}+\textcolor{green}{c_{36}+d_{36}}$ & $a_{39}+\textcolor{green}{c_{39}+d_{39}}$ \\
& $a_{31}+\textcolor{green}{c_{31}+d_{31}}$ & $a_{34}+\textcolor{green}{c_{34}+d_{34}}$ & $a_{37}+\textcolor{green}{c_{37}+d_{37}}$ & $a_{40}+\textcolor{green}{c_{40}+d_{40}}$ \\
& $\textcolor{red}{b_{29}}+\textcolor{green}{c_{41}+d_{41}}$ & $\textcolor{red}{b_{32}}+\textcolor{green}{c_{44}+d_{44}}$ & $\textcolor{red}{b_{35}}+\textcolor{green}{c_{47}+d_{47}}$ & $\textcolor{red}{b_{38}}+\textcolor{green}{c_{50}+d_{50}}$ \\
& $\textcolor{red}{b_{30}}+\textcolor{green}{c_{42}+d_{42}}$ & $\textcolor{red}{b_{33}}+\textcolor{green}{c_{45}+d_{45}}$ & $\textcolor{red}{b_{36}}+\textcolor{green}{c_{48}+d_{48}}$ & $\textcolor{red}{b_{39}}+\textcolor{green}{c_{51}+d_{51}}$ \\
& $\textcolor{red}{b_{31}}+\textcolor{green}{c_{43}+d_{43}}$ & $\textcolor{red}{b_{34}}+\textcolor{green}{c_{46}+d_{46}}$ & $\textcolor{red}{b_{37}}+\textcolor{green}{c_{49}+d_{49}}$ & $\textcolor{red}{b_{40}}+\textcolor{green}{c_{52}+d_{52}}$ \\
& $a_{41}+\textcolor{red}{b_{41}}+\textcolor{green}{c_{53}+d_{53}}$ & $a_{44}+\textcolor{red}{b_{44}}+\textcolor{green}{c_{56}+d_{56}}$ & $a_{47}+\textcolor{red}{b_{47}}+\textcolor{green}{c_{59}+d_{59}}$ & $a_{50}+\textcolor{red}{b_{50}}+\textcolor{green}{c_{62}+d_{62}}$ \\
& $a_{42}+\textcolor{red}{b_{42}}+\textcolor{green}{c_{54}+d_{54}}$ & $a_{45}+\textcolor{red}{b_{45}}+\textcolor{green}{c_{57}+d_{57}}$ & $a_{48}+\textcolor{red}{b_{48}}+\textcolor{green}{c_{60}+d_{60}}$ & $a_{51}+\textcolor{red}{b_{51}}+\textcolor{green}{c_{63}+d_{63}}$ \\
& $a_{43}+\textcolor{red}{b_{43}}+\textcolor{green}{c_{55}+d_{55}}$ & $a_{46}+\textcolor{red}{b_{46}}+\textcolor{green}{c_{58}+d_{58}}$ & $a_{49}+\textcolor{red}{b_{49}}+\textcolor{green}{c_{61}+d_{61}}$ & $a_{52}+\textcolor{red}{b_{52}}+\textcolor{green}{c_{64}+d_{64}}$ \\
& $a_{53}+\textcolor{red}{b_{53}}$ & $a_{56}+\textcolor{red}{b_{56}}$ & $a_{59}+\textcolor{red}{b_{59}}$ & $a_{62}+\textcolor{red}{b_{62}}$ \\
& $a_{54}+\textcolor{red}{b_{54}}$ & $a_{57}+\textcolor{red}{b_{57}}$ & $a_{60}+\textcolor{red}{b_{60}}$ & $a_{63}+\textcolor{red}{b_{63}}$ \\
& $a_{55}+\textcolor{red}{b_{55}}$ & $a_{58}+\textcolor{red}{b_{58}}$ & $a_{61}+\textcolor{red}{b_{61}}$ & $a_{64}+\textcolor{red}{b_{64}}$ \\
\hline
\end{tabular}}
\end{table}

\begin{table}[h]
\centering
\caption{NB Scheme in $(N, T_1: K_1, T_2:K_2) = (4, 2:2, 1:4)$ for retrieving $W_4$} \label{tbl:eg-NB-d}
\resizebox{0.98\textwidth}{!}{
\begin{tabular}{|c||c|c|c|c|c|}\hline
Coding group&\mbox{Server-1}&\mbox{Server-2}&\mbox{Server-3} & \mbox{Server-4}\\
\hline
$\textcolor{Mulberry}{a}: (32, 8)$ & $\textcolor{Mulberry}{a_1}, \textcolor{red}{b_1}$ & $\textcolor{Mulberry}{a_3}, \textcolor{red}{b_3}$ & $\textcolor{Mulberry}{a_5}, \textcolor{red}{b_5}$ & $\textcolor{Mulberry}{a_7}, \textcolor{red}{b_7}$ \\
$\textcolor{red}{b}: (32, 8)$ & $\textcolor{Mulberry}{a_2}, \textcolor{red}{b_2}$ & $\textcolor{Mulberry}{a_4}, \textcolor{red}{b_4}$ & $\textcolor{Mulberry}{a_6}, \textcolor{red}{b_6}$ & $\textcolor{Mulberry}{a_8}, \textcolor{red}{b_8}$ \\
\hline\hline
$\textcolor{blue}{a+b}: (32, 8)$& $\textcolor{blue}{a_{9} + b_{9}}$ & $\textcolor{blue}{a_{11} + b_{11}}$ & $\textcolor{blue}{a_{13} + b_{13}}$ & $\textcolor{blue}{a_{15} + b_{15}}$ \\
& $\textcolor{blue}{a_{10} + b_{10}}$ & $\textcolor{blue}{a_{12} + b_{12}}$ & $\textcolor{blue}{a_{14} + b_{14}}$ & $\textcolor{blue}{a_{16} + b_{16}}$ \\
\hline\hline\hline
$d$: (64, 64) & $\textcolor{green}{c_1}, d_1$ & $\textcolor{green}{c_2}, d_2$ & $\textcolor{green}{c_3}, d_3$ & $\textcolor{green}{c_4}, d_4$ \\
\hline\hline
$\textcolor{green}{c}: (64, 16)$ & $\textcolor{green}{c_5}+d_5$ & $\textcolor{green}{c_8}+d_8$ & $\textcolor{green}{c_{11}}+d_{11}$ & $\textcolor{green}{c_{14}}+d_{14}$ \\
& $\textcolor{green}{c_6}+d_6$ & $\textcolor{green}{c_9}+d_9$ & $\textcolor{green}{c_{12}}+d_{12}$ & $\textcolor{green}{c_{15}}+d_{15}$ \\
& $\textcolor{green}{c_7}+d_7$ & $\textcolor{green}{c_{10}}+d_{10}$ & $\textcolor{green}{c_{13}}+d_{13}$ & $\textcolor{green}{c_{16}}+d_{16}$ \\
\hline\hline\hline
 & $\textcolor{Mulberry}{a_{17}}+\textcolor{green}{c_{17}}$, $\textcolor{red}{b_{17}}+d_{17}$ & $\textcolor{Mulberry}{a_{18}}+\textcolor{green}{c_{18}}$, $\textcolor{red}{b_{18}}+d_{18}$ & $\textcolor{Mulberry}{a_{19}}+\textcolor{green}{c_{19}}$, $\textcolor{red}{b_{19}}+d_{19}$ & $\textcolor{Mulberry}{a_{20}}+\textcolor{green}{c_{20}}$, $\textcolor{red}{b_{20}}+d_{20}$ \\
 & $\textcolor{Mulberry}{a_{21}}+\textcolor{green}{c_{21}}$, $\textcolor{red}{b_{21}}+d_{21}$ & $\textcolor{Mulberry}{a_{22}}+\textcolor{green}{c_{22}}$, $\textcolor{red}{b_{22}}+d_{22}$ & $\textcolor{Mulberry}{a_{23}}+\textcolor{green}{c_{23}}$, $\textcolor{red}{b_{23}}+d_{23}$ & $\textcolor{Mulberry}{a_{24}}+\textcolor{green}{c_{24}}$, $\textcolor{red}{b_{24}}+d_{24}$ \\
& $\textcolor{Mulberry}{a_{25}}+\textcolor{green}{c_{25}}$, $\textcolor{red}{b_{25}}+d_{25}$ & $\textcolor{Mulberry}{a_{26}}+\textcolor{green}{c_{26}}$, $\textcolor{red}{b_{26}}+d_{26}$ & $\textcolor{Mulberry}{a_{27}}+\textcolor{green}{c_{27}}$, $\textcolor{red}{b_{27}}+d_{27}$ & $\textcolor{Mulberry}{a_{28}}+\textcolor{green}{c_{28}}$, $\textcolor{red}{b_{28}}+d_{28}$ \\
& $\textcolor{Mulberry}{a_{29}}+\textcolor{green}{c_{29}}+d_{29}$ & $\textcolor{Mulberry}{a_{32}}+\textcolor{green}{c_{32}}+d_{32}$ & $\textcolor{Mulberry}{a_{35}}+\textcolor{green}{c_{35}}+d_{35}$ & $\textcolor{Mulberry}{a_{38}}+\textcolor{green}{c_{38}}+d_{38}$ \\
& $\textcolor{Mulberry}{a_{30}}+\textcolor{green}{c_{30}}+d_{30}$ & $\textcolor{Mulberry}{a_{33}}+\textcolor{green}{c_{33}}+d_{33}$ & $\textcolor{Mulberry}{a_{36}}+\textcolor{green}{c_{36}}+d_{36}$ & $\textcolor{Mulberry}{a_{39}}+\textcolor{green}{c_{39}}+d_{39}$ \\
& $\textcolor{Mulberry}{a_{31}}+\textcolor{green}{c_{31}}+d_{31}$ & $\textcolor{Mulberry}{a_{34}}+\textcolor{green}{c_{34}}+d_{34}$ & $\textcolor{Mulberry}{a_{37}}+\textcolor{green}{c_{37}}+d_{37}$ & $\textcolor{Mulberry}{a_{40}}+\textcolor{green}{c_{40}}+d_{40}$ \\
& $\textcolor{red}{b_{29}}+\textcolor{green}{c_{41}}+d_{41}$ & $\textcolor{red}{b_{32}}+\textcolor{green}{c_{44}}+d_{44}$ & $\textcolor{red}{b_{35}}+\textcolor{green}{c_{47}}+d_{47}$ & $\textcolor{red}{b_{38}}+\textcolor{green}{c_{50}}+d_{50}$ \\
& $\textcolor{red}{b_{30}}+\textcolor{green}{c_{42}}+d_{42}$ & $\textcolor{red}{b_{33}}+\textcolor{green}{c_{45}}+d_{45}$ & $\textcolor{red}{b_{36}}+\textcolor{green}{c_{48}}+d_{48}$ & $\textcolor{red}{b_{39}}+\textcolor{green}{c_{51}}+d_{51}$ \\
& $\textcolor{red}{b_{31}}+\textcolor{green}{c_{43}}+d_{43}$ & $\textcolor{red}{b_{34}}+\textcolor{green}{c_{46}}+d_{46}$ & $\textcolor{red}{b_{37}}+\textcolor{green}{c_{49}}+d_{49}$ & $\textcolor{red}{b_{40}}+\textcolor{green}{c_{52}}+d_{52}$ \\
& $\textcolor{blue}{a_{41}+b_{41}}+\textcolor{green}{c_{53}}+d_{53}$ & $\textcolor{blue}{a_{44}+b_{44}}+\textcolor{green}{c_{56}}+d_{56}$ & $\textcolor{blue}{a_{47}+b_{47}}+\textcolor{green}{c_{59}}+d_{59}$ & $\textcolor{blue}{a_{50}+b_{50}}+\textcolor{green}{c_{62}}+d_{62}$ \\
& $\textcolor{blue}{a_{42}+b_{42}}+\textcolor{green}{c_{54}}+d_{54}$ & $\textcolor{blue}{a_{45}+b_{45}}+\textcolor{green}{c_{57}}+d_{57}$ & $\textcolor{blue}{a_{48}+b_{48}}+\textcolor{green}{c_{60}}+d_{60}$ & $\textcolor{blue}{a_{51}+b_{51}}+\textcolor{green}{c_{63}}+d_{63}$ \\
& $\textcolor{blue}{a_{43}+b_{43}}+\textcolor{green}{c_{55}}+d_{55}$ & $\textcolor{blue}{a_{46}+b_{46}}+\textcolor{green}{c_{58}}+d_{58}$ & $\textcolor{blue}{a_{49}+b_{49}}+\textcolor{green}{c_{61}}+d_{61}$ & $\textcolor{blue}{a_{52}+b_{52}}+\textcolor{green}{c_{64}}+d_{64}$ \\
& $\textcolor{blue}{a_{53}+b_{53}}$ & $\textcolor{blue}{a_{56}+b_{56}}$ & $\textcolor{blue}{a_{59}+b_{59}}$ & $\textcolor{blue}{a_{62}+b_{62}}$ \\
& $\textcolor{blue}{a_{54}+b_{54}}$ & $\textcolor{blue}{a_{57}+b_{57}}$ & $\textcolor{blue}{a_{60}+b_{60}}$ & $\textcolor{blue}{a_{63}+b_{63}}$ \\
& $\textcolor{blue}{a_{55}+b_{55}}$ & $\textcolor{blue}{a_{58}+b_{58}}$ & $\textcolor{blue}{a_{61}+b_{61}}$ & $\textcolor{blue}{a_{64}+b_{64}}$ \\
\hline
\end{tabular}}
\end{table}

\vspace{0.2cm}
\noindent\textbf{Encoding:} 
The coding structure is illustrated in Table \ref{tbl:eg-NB-a} and Table \ref{tbl:eg-NB-d}, for the retrieval of $W_1$ and $W_4$, respectively. The coding procedure also consists of three steps, as in the NS code, however the patterns are different, which is evident from the tables.

\vspace{0.2cm}
\noindent\textbf{Decoding and correctness:} There are three blocks in Table \ref{tbl:eg-NB-a} and Table \ref{tbl:eg-NB-d}. In Table \ref{tbl:eg-NB-a}, the symbols $c_{1:4}$, $d_{1:4}$, and $c_{5:16}+d_{5:16}$ in the second block can be used to reconstruct the interference signals in the third block, i.e., $c_{17:28}$, $d_{17:28}$, and $c_{29:64} + d_{29:64}$, by the property of the MDS code in each coding group. Canceling these interference signals generated by $W_{3:4}$, i.e., eliminating the coded symbols $c$ and $d$, Table \ref{tbl:eg-NB-a} essentially reduces to the scheme discussed in Section \ref{sec:review} for the $2$-private system: here $32$ interference signals $b_{1:8, 17:40}$ can be used to reconstruct $b_{9:16, 41:64}$. The desired message $W_1$ can thus be recovered. By symmetry, $W_2$ can be retrieved similarly. 

In Table \ref{tbl:eg-NB-d}, the symbols $a_{1:8}$, $b_{1:8},$ and $a_{9:16}+b_{9:16}$ in the first block can be used to reconstruct the interference signals in the third block, i.e., $a_{17:40}$, $b_{17:40}$, and $a_{41:64} + b_{41:64}$. Canceling the interference signals generated by $W_{1:2}$, i.e., eliminating the coded symbols $a$ and $b$, Table \ref{tbl:eg-NB-d} reduces to the scheme discussed in Section \ref{sec:review} for the $1$-private system, and the desired message $W_4$ can be recovered. By symmetry, $W_3$ can be retrieved similarly. 

\vspace{0.2cm}
\noindent\textbf{Privacy:}  When message $W_1$ or $W_2$ is requested, the coded symbols of message $W_{3:4}$ are downloaded as interference signals, and the interference signals such as $c$, $d$, or $c+d$ are mixed to $a$, $b$, $a+b$. With the symbols $c, d$ eliminated in Table \ref{tbl:eg-NB-a}, we have the retrieval pattern of the $2$-private system, which is clearly $2$-private. To see all the messages have privacy level $1$, observe that the coding pattern is the same for the retrieval of any message. In both Table \ref{tbl:eg-NB-a} and Table \ref{tbl:eg-NB-d}, the coding matrix of the coded symbols for any single message from any single server has full row rank. Thus by Lemma \ref{lem:equivlent}, messages $W_{1:4}$ have privacy level $1$.

\vspace{0.2cm}
\noindent\textbf{Performance:} The total number of downloaded symbols is $116$ and the message length is $64$. Thus the rate is $R_{\NB} = \frac{64}{116} = \frac{16}{29}$, which coincides with $R_{\NS}$ in this example. 
\vspace{0.2cm}

\noindent\textit{Remark:} The coding structure has the following feature: eliminating the coded symbols $c$ and $d$ in Table \ref{tbl:eg-NB-a} or Table \ref{tbl:eg-NB-d}, the remaining part has the same coding structure as the $2$-private code discussed in Section \ref{sec:review}; eliminating the coded symbols $a$ and $b$ in Table \ref{tbl:eg-NB-a} or Table \ref{tbl:eg-NB-d}, the remaining part has the same coding structure as the $1$-private code discussed in Section \ref{sec:review}. The NB coding structure can be interpreted as a mixture of the $T_1$-private code of message $W_{1:K_1}$ and $T_2$-private code for messages $W_{K_1+1:K_2}$ discussed in Section \ref{sec:review}, which is constructed in three blocks. Since the retrieval needs to follow the same pattern, the underlying $T_1$-private code and the underlying $T_2$-private code are required to have the same message length.  A $\frac{T_2}{N}$ fraction of the $T_1$-private code forms the first block, a $\frac{T_2}{N}$ fraction of the $T_2$-private code forms the second block. The remaining $\frac{N - T_2}{N}$ fractions of both codes are mixed together to form the third block by simple pairwise summations in an arbitrary order; in case they have different numbers of remaining coded symbols, the remaining summands are included directly.

\section{Capacity Upper Bounds} \label{sec:upper-bound}

\subsection{Converse proof of Theorem \ref{thm:main}}
Similar to the converse proof of $T$-private systems in \cite{sun2017capacity}, we first introduce the following lemma, and apply the lemma iteratively to provide an upper bound of the capacity.
\begin{lemma}\label{lem:iter}
For any $i \in 1:K_1-1$,
\begin{align}
H(A_{1:N}^{[i]}  | Q_{1:N}^{[i]}, W_{1:i} ) \geq \frac{T_1}{N}L + \frac{T_1}{N} H(A_{1:N}^{[i+1]}  | Q_{1:N}^{[i+1]}, W_{1:i+1});
\end{align}
and for any $j \in K_1:K_2 - 1$,
\begin{align}
H(A_{1:N}^{[j]}  | Q_{1:N}^{[j]}, W_{1:j} ) \geq \frac{T_2}{N}L + \frac{T_2}{N} H(A_{1:N}^{[j+1]}  | Q_{1:N}^{[j+1]}, W_{1:j+1}).
\end{align}
\end{lemma}
\begin{proof}[Proof of Lemma \ref{lem:iter}]
For any $i \in 1:K_1-1$,
\begin{align}
&H(A_{1:N}^{[i]}  | Q_{1:N}^{[i]}, W_{1:i} ) \\
&\geq \binom{N}{T_1}^{-1} \sum_{\Tc_1: |\Tc_1| = T_1} H(A_{\Tc_1}^{[i]}  | Q_{1:N}^{[i]}, W_{1:i} ) = \binom{N}{T_1}^{-1} \sum_{\Tc_1: |\Tc_1| = T_1} H(A_{\Tc_1}^{[i]}  | Q_{\Tc_1}^{[i]}, W_{1:i} ) \\
&= \binom{N}{T_1}^{-1} \sum_{\Tc_1: |\Tc_1| = T_1} H(A_{\Tc_1}^{[i+1]}  | Q_{\Tc_1}^{[i+1]}, W_{1:i} ) = \binom{N}{T_1}^{-1} \sum_{\Tc_1: |\Tc_1| = T_1} H(A_{\Tc_1}^{[i+1]}  | Q_{1:N}^{[i+1]}, W_{1:i} ) \\
&\geq \frac{T_1}{N} H(A_{1:N}^{[i+1]} | Q_{1:N}^{[i+1]}, W_{1:i}) = \frac{T_1}{N} H(A_{1:N}^{[i+1]}, W_{i+1} | Q_{1:N}^{[i+1]}, W_{1:i}) \\
&= \frac{T_1}{N} L + \frac{T_1}{N} H(A_{1:N}^{[i+1]} | Q_{1:N}^{[i+1]}, W_{1:i+1}),
\end{align}
where the first and third equalities are due to the Markov string $Q_{1:N \slash \Tc}^{[k]} \leftrightarrow Q_{\Tc}^{[k]} \leftrightarrow (W_{1:K_2}, A_{\Tc}^{[k]})$ and the independence between $W_{1:K_2}$ and $Q_{1:N}^{[k]}$ for any $k \in 1:K_2$ and $\Tc \subset 1:N$; the second equality is by the $T_1$-privacy among messages $W_{1:K_1}$; and the last inequality is by Han's inequality. Similarly, for any $j \in K_1 : K_2-1$, 
\begin{align}
&H(A_{1:N}^{[j]} | Q_{1:N}^{[j]}, W_{1:j} ) \\
&\geq \binom{N}{T_2}^{-1} \sum_{\Tc_2: |\Tc_2| = T_2} H(A_{\Tc_2}^{[j]} | Q_{1:N}^{[j]}, W_{1:j} ) = \binom{N}{T_2}^{-1} \sum_{\Tc_2: |\Tc_2| = T_2} H(A_{\Tc_2}^{[j]} | Q_{\Tc_2}^{[j]}, W_{1:j} ) \\
&= \binom{N}{T_2}^{-1} \sum_{\Tc_2: |\Tc_2| = T_2} H(A_{\Tc_2}^{[j+1]} | Q_{\Tc_2}^{[j+1]}, W_{1:j} ) = \binom{N}{T_2}^{-1} \sum_{\Tc_2: |\Tc_2| = T_2} H(A_{\Tc_2}^{[j+1]} | Q_{1:N}^{[j+1]}, W_{1:j} ) \\
&\geq \frac{T_2}{N} H(A_{1:N}^{[j+1]} | Q_{1:N}^{[j+1]}, W_{1:j}) = \frac{T_2}{N} H(A_{1:N}^{[j+1]}, W_{j+1} | Q_{1:N}^{[j+1]}, W_{1:j}) \\
&= \frac{T_2}{N} L + \frac{T_2}{N} H(A_{1:N}^{[j+1]} | Q_{1:N}^{[j+1]}, W_{1:j+1}).
\end{align}
The lemma is proved.
\end{proof}
With the inequalities in Lemma \ref{lem:iter}, we prove the $\overline{R}$ upper bound in Theorem \ref{thm:main}.
\begin{align}
&\sum_{n = 1}^N H(A^{[1]}_{n} | Q_{1:N}^{[1]}) \geq H(A^{[1]}_{1:N} | Q_{1:N}^{[1]}) = H(W_1) + H(A^{[1]}_{1:N}  | Q_{1:N}^{[1]}, W_1)
\end{align}
Iteratively apply Lemma \ref{lem:iter} by letting $k = 1, 2, \ldots,  K_1 - 1$,
\begin{align}
H(A^{[1]}_{1:N}  | Q_{1:N}^{[1]}, W_1) &\geq \left( \frac{T_1}{N} + \cdots + \left(\frac{T_1}{N}\right)^{K_1 - 1} \right) L + \left(\frac{T_1}{N}\right)^{K_1 - 1} H(A_{1:N}^{[K_1]} | Q_{1:N}^{[K_1]}, W_{1:K_1})
\end{align}
Then iteratively apply Lemma \ref{lem:iter} with $k = K_1, \ldots, K_2 -1$,
\begin{align}
H(A_{1:N}^{[K_1]} | Q_{1:N}^{[K_1]}, W_{1:K_1}) & \geq \left(\frac{T_2}{N} + \cdots + \left(\frac{T_2}{N}\right)^{K_2-K_1} \right)L.
\end{align}
Therefore, we have
\begin{align}
&\sum_{n = 1}^N H(A^{[1]}_{n} | Q_{1:N}^{[1]}) \geq \left(1 + \frac{T_1}{N}+ \cdots + \left(\frac{T_1}{N}\right)^{K_1 - 1} \right) L +\left(\frac{T_1}{N}\right)^{K_1 - 1} \left(\frac{T_2}{N} + \cdots + \left(\frac{T_2}{N}\right)^{K_2-K_1} \right)L \notag \\
&= D_N^*\left(K_1, T_1\right)L +  \frac{T_2}{N} \left( \frac{T_1}{N} \right)^{K_1-1} D_N^*(K_2 - K_1, T_2)L.
\end{align}
It follows 
\begin{align}
& C \leq \frac{L}{\sum_{n = 1}^N H(A_n^{[1]} | Q_{1:N}^{[1]}) } \leq \left( D_N^*\left(K_1, T_1\right) +  \frac{T_2}{N} \left( \frac{T_1}{N} \right)^{K_1-1} D_N^*(K_2 - K_1, T_2)  \right)^{-1},
\end{align}
which concludes the proof. \qed

\subsection{Proof of Proposition \ref{lem:3,2:2,1:3}}
First, we have
\begin{align}
&H(A_{1:3}^{[1]} | Q_{1:3}^{[1]}, W_1) \geq \frac{1}{2}\left( H(A_{1,2}^{[1]} | Q_{1,2}^{[1]}, W_1) + H(A_{1,3}^{[1]} | Q_{1,3}^{[1]}, W_1) \right) \\
& = \frac{1}{2}\left( H(A_{1,2}^{[2]} | Q_{1,2}^{[2]}, W_1) + H(A_{1,3}^{[2]} | Q_{1,3}^{[2]}, W_1) \right) \\
& = \frac{1}{2}\left( H(A_{1,2}^{[2]} | Q_{1:3}^{[2]}, W_1) + H(A_{1,3}^{[2]} | Q_{1:3}^{[2]}, W_1) \right) \\
& \geq \frac{1}{2}\left( H(A_{1}^{[2]} | Q_{1:3}^{[2]}, W_1) + H(A_{1:3}^{[2]} | Q_{1:3}^{[2]}, W_1) \right) \\
& = \frac{1}{2}\left( H(A_{1}^{[2]} | Q_{1:3}^{[2]}, W_1) + L + H(A_{1:3}^{[2]} | Q_{1:3}^{[2]}, W_{1:2}) \right) \\
& \geq \frac{1}{2} \left( H(A_{1}^{[2]} | Q_{1:3}^{[2]}, W_1) + L + \frac{1}{3}L + \frac{1}{3} H(A_{1:3}^{[3]} | Q_{1:3}^{[3]}, W_{1:3}) \right)\\
& = \frac{1}{2} H(A_{1}^{[2]} | Q_{1}^{[2]}, W_1) + \frac{2}{3} L,
\end{align}
where the first equality is because $W_{1:2}$ have privacy level $2$; the second inequality is the submodular inequality; and the last inequality is by Lemma \ref{lem:iter}. By the symmetry between messages $W_{1:2}$ and symmetries among servers $1:3$, we can similarly derive for any $n = 1,2,3$,
\begin{align}
H(A_{1:3}^{[1]} | Q_{1:3}^{[1]}, W_1) &\geq \frac{1}{2} H(A_{n}^{[2]} | Q_{n}^{[2]}, W_1) + \frac{2}{3}L, \\
H(A_{1:3}^{[2]} | Q_{1:3}^{[2]}, W_2) & \geq \frac{1}{2}H(A_{n}^{[1]} | Q_{n}^{[1]}, W_2) + \frac{2}{3}L.
\end{align}
The summation of the first term in each lower bound above satisfies
\begin{align}
&H(A_{n}^{[1]} | Q_{n}^{[1]}, W_2) +  H(A_{n}^{[2]} | Q_{n}^{[2]}, W_1) = H(A_{n}^{[1]}, W_2 | Q_{n}^{[1]}) + H(A_{n}^{[2]}, W_1 | Q_{n}^{[2]}) - 2L \\
& = H(A_{n}^{[1]}, W_2 | Q_{n}^{[1]}) + H(A_{n}^{[1]}, W_1 | Q_{n}^{[1]}) - 2L \\
& \geq H(A_{n}^{[1]} | Q_{n}^{[1]}) + H(A_{n}^{[1]}, W_{1,2} | Q_{n}^{[1]}) - 2L = H(A_{n}^{[1]} | Q_{n}^{[1]})  + H(A_{n}^{[1]} | Q_{n}^{[1]}, W_{1,2})\\
& \geq H(A_{n}^{[1]} | Q_{1:3}^{[1]})  + \frac{1}{3}L,
\end{align}
where the first inequality is the submodular inequality; and the last inequality is by Lemma \ref{lem:iter}.
Then it follows that
\begin{align}
H(A_{1:3}^{[1]} | Q_{1:3}^{[1]}, W_1) + H(A_{1:3}^{[2]} | Q_{1:3}^{[2]}, W_2) \geq \frac{1}{6} \sum_{n = 1}^3 H(A_{n}^{[1]} | Q_{1:3}^{[1]}) + \frac{3}{2}L. \label{eqn:sub}
\end{align}
By the inequality (\ref{eqn:sub}) above, we have
\begin{align}
& 2 \sum_{n = 1}^3 H(A_{n}^{[1]} | Q_{1:3}^{[1]}) =  \sum_{n = 1}^3 H(A_{n}^{[1]} | Q_{1:3}^{[1]}) + \sum_{n = 1}^3 H(A_{n}^{[2]} | Q_{1:3}^{[2]}) \\
& \geq H(A_{1:3}^{[1]} | Q_{1:3}^{[1]}) + H(A_{1:3}^{[2]} | Q_{1:3}^{[2]}) \\
& \geq  H(A_{1:3}^{[1]} | Q_{1:3}^{[1]}, W_1) + H(A_{1:3}^{[2]} | Q_{1:3}^{[2]}, W_2) + 2L\\
& \geq \frac{1}{6} \sum_{n = 1}^3 H(A_{n}^{[1]} | Q_{1:3}^{[1]}) + \frac{7}{2}L.
\end{align}
It implies
\begin{align}
\sum_{n = 1}^3 H(A_{n}^{[1]} | Q_{1:3}^{[1]})  \geq \frac{21}{11}L,
\end{align}
which gives that
\begin{equation*}
C \leq \frac{L}{\sum_{n = 1}^3 H(A_{n}^{[1]} | Q_{1:3}^{[1]}) } \leq \frac{11}{21}. \qedhere
\end{equation*}
\qed

\section{The Non-uniform Successive Cancellation Scheme}  \label{sec:NS}

In this section, we provide the general code construction for the non-uniform successive cancellation scheme. 

\subsection{Specifying coding group parameters}

It is clear from the example in Section \ref{sec:eg-NS} that the proposed code can be viewed as consisting of $K_2$ layers and multiple coding groups. We next first specify the appropriate parameters for each coding group. We identify each coding group by its composition. For example, in Table \ref{tbl:eg-NS-a}, the red coding group is of form \textcolor{red}{$b+c$}, and thus we can use the set of message indices involved to identify it as $\mathcal{K}=\{2,3\}$, i.e., it involves the messages $(W_2,W_3)$. Clearly this coding group will be placed in the $2^{nd}$ and $3^{rd}$ layers. 

More generally, for each coding group, there are a total of five parameters to specify: the total number of coded symbols $n_1(\Kc)$ and $n_2(\Kc)$, and the number of MDS code message symbols $k_1(\Kc)$ and $k_2(\Kc)$ when retrieving a message of privacy level $T_1$ and that of privacy level $T_2$, respectively; and the number of symbols to be placed in the top layer $m(\Kc)$.
In other words, during the retrieval of a message $W_{k^*}$, when $k^*\in 1:K_1$, an $(n_1(\Kc), k_1(\Kc))$ MDS code is used for this coding group, while during the retrieval of a message $W_{k^*}$, when $k^*\in K_1+1:K_2$, an $(n_2(\Kc), k_2(\Kc))$ MDS code is used for this coding group. After MDS encoding, $m(\Kc)$ symbols will be placed in the $|\Kc|$-th layer, while the remaining will be placed in the $(|\Kc|+1)$-th layer as interference, and the symbols are uniformly distributed across all servers.

The message length for the NS coding scheme is $L = N^{K_2}$ in the proposed scheme; note that the length may be reduced in some cases, however we choose this value to simplify the presentation of the code construction without any loss in terms of the download cost. To introduce $(n_1(\Kc),k_1(\Kc), n_2(\Kc), k_2(\Kc),m(\Kc))$, we first define 
\begin{align}
M \triangleq T_2^{K_2-K_1} + \frac{T_1 - T_2}{N - T_2}\left( N^{K_2-K_1} - T_2^{K_2-K_1} \right) = N^{K_2-K_1} - \frac{N - T_1}{N - T_2}\left( N^{K_2-K_1} - T_2^{K_2-K_1} \right),
\end{align}
which is an integer. For any $(i, j) \in 0:K_1 \times 0:K_2-K_1$, define $d_{0, 0} \triangleq 0$ and for $i + j \geq 1$, define
\[
d_{i,j} \triangleq \left\{
\begin{array}{ll}
 M T_1^{K_1 - i} (N - T_1)^{i - 1}, \quad \text{if }j = 0\\
 T_1^{K_1 - i} (N - T_1)^{i} T_2^{K_2 - K_1 - j} (N - T_2)^{j-1}, \quad \text{otherwise}
\end{array}
\right.
\]
Then we specify 
\begin{align}
m(\Kc) &\triangleq N d_{|\Kc \cap 1:K_1|, |\Kc \cap K_1+1:K_2|},
\end{align}
and
\begin{align}
n_1(\Kc) &\triangleq m(\Kc) + N d_{|\Kc \cap 1:K_1| +1, |\Kc \cap K_1+1:K_2|}, \quad k_1(\Kc) \triangleq \frac{T_1}{N} n_1(\Kc),\\
n_2(\Kc) &\triangleq m(\Kc) + N d_{|\Kc \cap 1:K_1|, |\Kc \cap K_1+1:K_2|+1}, \quad k_2(\Kc) \triangleq \frac{T_2}{N} n_2(\Kc).
\end{align}

The properties of the functions used for encoding, correctness and privacy of the NS coding scheme, are summarized as Lemma \ref{lem:property-mnk} below, which is proved in the appendix.
\begin{lemma}\label{lem:property-mnk}
The tuple $(n_1(\cdot), k_1(\cdot), n_2(\cdot), k_2(\cdot), m(\cdot))$ has the following properties:
\begin{enumerate}
\item For any non-empty $\Kc \subset 1:K_2$,
\begin{align}
k_1(\Kc) = m(\Kc), \quad k_2(\Kc) \leq m(\Kc)
\end{align}
\item The following equality holds:
\begin{align}
\sum_{\Kc \subset 1:K_2,~k^* \in \Kc} m(\Kc) = L
\end{align}
\item When $k^* \in 1:K_1$, for any $k \not= k^*$ the following inequality holds:
\begin{align}
\sum_{\Kc \subset 1:K_2\slash \{k^*\},~ k \in \Kc} k_1(\Kc) < L
\end{align}
When $k^* \in K_1+1:K_2$, for any $k\not = k^*$, the following inequality holds:
\begin{align}
\sum_{\Kc \subset 1:K_2\slash \{k^*\},~k \in \Kc} k_2(\Kc) < L
\end{align}
\end{enumerate}
\end{lemma}

\subsection{Encoding, decoding, privacy, and performance}

\noindent\textbf{Encoding:} The queries and answers are formed in three steps, and the queries are simply the encoding matrix for these answers. Assume for each $(n,k)$ pair where $n\geq k$, an MDS code in $\Fb_q$ is given and fixed, and we refer to it as the $(n,k)$ MDS code.
\begin{enumerate}
\item \textit{Precoding:} Let $S_{1:K_2}$ be $K_2$ independent random matrices, which are uniformly drawn from the set of all $N^{K_2} \times N^{K_2}$ full rank matrices over $\Fb_q$; these matrices are known only to the user. The precoded messages $W^*_{1:K_2}$ are
\begin{align}
W^*_k = S_k W_k, \quad \forall k \in 1:K_2.
\end{align}
\item \textit{Group-wise MDS coding:}

The precoded messages are partitioned into non-overlapping segments, and each segment is MDS-coded under certain parameters. {We use $W_k(\Kc)$ to denote a segment of message $W_k$ indexed by $\Kc \subset 1:K_2$.} One special coding group corresponds to the precoded desired message $W^*_{k^*}$, where the precoded message $W^*_{k^*}$ is $(N^{K_2}, N^{K_2})$ MDS-coded into $\tilde{W}_{k^*}$, which is then partitioned into non-overlapping segments $\tilde{W}_{k^*}(\Kc \cup \{k^*\})$ for each $\Kc \subset 1:K_2$, where $\tilde{W}_{k^*}(\Kc \cup \{k^*\})$ has length $m(\Kc \cup \{k^*\})$. The non-overlapping segments of $W^*_{k^*}$ exist because of item 2 in Lemma \ref{lem:property-mnk}. Other coding groups are indexed by non-empty sets $\Kc \subset 1:K_2 \slash \{k^*\}$.
For each $\Kc \subset 1:K_2 \slash \{k^*\}$, the coding group indexed by $\Kc$ is specified as follows.
\begin{itemize}
\item If $k^* \in 1:K_1$, for each $k \in \Kc$, a segment of $W^*_k$ with length $k_1(\Kc)$ is $(n_1(\Kc), k_1(\Kc))$ MDS-coded into $(\tilde{W}_k(\Kc), \tilde{W}_k(\Kc \cup \{k^*\}))$, which have lengths $m(\Kc)$ and $m(\Kc \cup \{k^*\})$, respectively.
\item If $k^* \in K_1+1:K_2$, for each $k \in \Kc$, a segment of $W^*_k$ with length $k_2(\Kc)$ is $(n_2(\Kc), k_2(\Kc))$ MDS-coded into $(\tilde{W}_k(\Kc), \tilde{W}_k(\Kc \cup \{k^*\}))$, which have lengths $m(\Kc)$ and $m(\Kc \cup \{k^*\})$, respectively.
\end{itemize}
The non-overlapping segments of $W^*_k$ for any $k \not= k^*$ exist because of item 3 in Lemma \ref{lem:property-mnk}.

\item \textit{Forming pre-coded message sums:} There are $K_2$ layers in the coding structure. The summation of the MDS-coded messages are placed in the layered structure from top to bottom as follows. 
For $i = 1,2,\ldots, K_2$, in the $i$-th layer, the summations (which are vectors) are
\begin{equation}
A_{1:N}^{[k^*]}(\Kc) = \sum_{k \in \Kc} \tilde{W}_k(\Kc), \quad \forall \Kc\subset 1:K_2 \text{ with } |\Kc| = i,
\end{equation}
and each vector is partitioned and distributed to $N$ servers uniformly. The MDS coded symbols of coding group indexed by $\Kc$ are shown in Table \ref{tbl:Kc-NS}.
\end{enumerate}

\begin{table}[tb!]
\centering
\caption{Placement of coding group indexed by $\Kc$ in the $|\Kc|$-th and $(|\Kc|+1)$-th layers} \label{tbl:Kc-NS}
\resizebox{\textwidth}{!}{
\begin{tabular}{|c|c||c|c|c|c|c|}\hline
Layer & Coding group&\mbox{Servers} $1:N$\\
\hline
$\vdots$ & $\vdots$ & $\vdots$  \\
\hline\hline
 & $\cdots$ & $\cdots$  \\
 $|\Kc|$-th & \textcolor{red}{coding group $\Kc$}: $(n_i(\Kc), k_i(\Kc))$, & $m(\Kc)$ symbols: $\textcolor{red}{\sum_{k \in \Kc} \tilde{W}_k(\Kc)}$ \\
 & where $i=1$ if $k^* \in 1:K_1$ otherwise $i=2$ & $\cdots$  \\
\hline\hline
 & $\cdots$ & $\cdots$  \\
 $(|\Kc|+1)$-th &  & $m(\Kc\cup \{k^*\})$ symbols: $\tilde{W}_{k^*}(\Kc) + \textcolor{red}{\sum_{k \in \Kc} \tilde{W}_k(\Kc \cup \{k^*\})}$ \\
 & $\cdots$ & $\cdots$  \\
\hline\hline
$\vdots$ & $\vdots$ & $\vdots$  \\
\hline
\end{tabular}}
\end{table}

\vspace{0.2cm}
\noindent\textbf{Decoding and correctness:} For any non-empty $\Kc \subset 1:K_2 \slash \{k^*\}$, the MDS-coded interference symbols $(\tilde{W}_k(\Kc), \tilde{W}_k(\Kc \cup \{k^*\}))_{k \in \Kc}$ in the coding group indexed by $\Kc$ are placed in two adjacent layers. Specifically, $(\tilde{W}_{k}(\Kc))_{k \in \Kc}$ are placed in the $|\Kc|$-th layer in the form of a signal
\begin{align}
A_{1:N}^{[k^*]}(\Kc) = \sum_{k \in \Kc} \tilde{W}_{k}(\Kc),
\end{align}
and $(\tilde{W}_{k}(\Kc \cup \{k^*\}))_{k \in \Kc}$ are placed in the $(|\Kc|+1)$-th layer in the form of
\begin{align}
A_{1:N}^{[k^*]}(\Kc \cup \{k^*\}) = \tilde{W}_{k^*}(\Kc \cup \{k^*\}) + \sum_{k \in \Kc} \tilde{W}_{k}(\Kc \cup \{k^*\}).
\end{align}
The interference signal in the top layer can cancel the interference signal in the bottom layer. The interference signal $\sum_{k \in \Kc} \tilde{W}_{k}(\Kc)$ in the $|\Kc|$-th layer (top layer) has length $m(\Kc)$.
\begin{itemize}
\item When $k^* \in 1:K_1$, since $(\tilde{W}_k(\Kc), \tilde{W}_k(\Kc \cup \{k^*\})$ are encoded by the same linear $(n_1(\Kc), k_1(\Kc))$ MDS code for each $k \in \Kc$, by item 1 in Lemma \ref{lem:property-mnk}, that $m(\Kc) = k_1(\Kc)$, the interference signal $\sum_{k \in \Kc} \tilde{W}_{k}(\Kc \cup \{k^*\})$ in the $(|\Kc|+1)$-th layer can indeed be recovered.
\item When $k^* \in K_1+1:K_2$, since $(\tilde{W}_k(\Kc), \tilde{W}_k(\Kc \cup \{k^*\})$ are encoded by the same linear $(n_2(\Kc), k_2(\Kc))$ MDS code for each $k \in \Kc$, by item 1 in Lemma \ref{lem:property-mnk}, that $m(\Kc) \geq k_2(\Kc)$, the interference signal $\sum_{k \in \Kc} \tilde{W}_{k}(\Kc \cup \{k^*\})$ in the $(|\Kc|+1)$-th layer can be recovered.
\end{itemize}
Thus we have $\tilde{W}_{k^*}(\Kc \cup \{k^*\})$ for all $\Kc \subset 1:K$, and the desired message $W_{k^*}$ can be recovered.

\vspace{0.2cm}
\noindent\textbf{Privacy:} The coding pattern, i.e., the manner of forming pre-coded message sums, is the same for the retrieval of any message $W_{k^*}$. Specifically, when the identity of the desired message $k^* \in 1:K_1$,
\begin{align}
n_1(\Kc) & = m(\Kc) + m(\Kc \cup \{k^*\}),
\end{align}
and when $k^* \in K_1:K_2$, 
\begin{align}
n_2(\Kc) & = m(\Kc) + m(\Kc \cup \{k^*\}).
\end{align}
Moreover, there are $m(\Kc)$ summations of form $\Kc$ placed in the $|\Kc|$-th layer. Thus the placements of the pre-coded message sums are the same for retrieving any message $W_{k^*}$. For example, there are $4$ sums of form $b+c$ in the $2^{nd}$ layers of both Table \ref{tbl:eg-NS-a} and Table \ref{tbl:eg-NS-d}. Similarly, the pre-coded sums can be indicated by the set of messages involved, e.g., summations of form $b+c$ are indicated by $\Kc=\{2, 3\}$. 

Since it is a linear code, the coded symbols can be generated by the corresponding coding matrices. When $k^* \in 1:K_1$, the desired precoded message $W^*_{k^*}$ is $(N^{K_2}, N^{K_2})$ MDS coded into $\tilde{W}_{k^*}$; and for each $k\not= k^*$, in the coding group $\Kc \subset 1:K_2 \slash \{k^*\}$ with $k \in \Kc$, 
a non-overlapping segment of $W^*_k$ is the $(n_1(\Kc), k_1(\Kc))$ MDS coded where $n_1(\Kc) : k_1(\Kc) = N:T_1$. Thus for any $k \in 1:K_2$ the coding matrix of MDS coded symbols $\tilde{W}_{k}$ in any $T_1$ servers from the segments of the precoded $W^*_k$ is a $T_1N^{K_2-1} \times T_1N^{K_2-1}$ full rank matrix. By applying Lemma \ref{lem:equivlent}, the messages $W_{1:K_1}$ thus have privacy level $T_1$. 

The statement above also implies that for any $k \in 1:K_2$ the coding matrix of MDS coded symbols $\tilde{W}_{k}$ in any $T_2$ servers from the segments of the precoded $W^*_k$ is a $T_2N^{K_2-1} \times T_2N^{K_2-1}$ full rank matrix. In addition, when $k^* \in K_1+1:K_2$, the desired precoded message $W^*_{k^*}$ is $(N^{K_2}, N^{K_2})$ MDS coded into $\tilde{W}_{k^*}$; and for each $k\not= k^*$, in the coding group $\Kc \subset 1:K_2 \slash \{k^*\}$ with $k \in \Kc$, 
a non-overlapping segment of $W^*_k$ is the $(n_2(\Kc), k_2(\Kc))$ MDS coded where $n_2(\Kc) : k_2(\Kc) = N:T_2$. Thus for any $k^* \in 1:K_2$, the coding matrix of MDS coded symbols $\tilde{W}_{k}$ in any $T_2$ servers from the segments of the precoded $W^*_k$ is a $T_2N^{K_2-1} \times T_2N^{K_2-1}$ full rank matrix. By applying Lemma \ref{lem:equivlent}, the messages $W_{1:K_2}$ thus have privacy level $T_2$. 

\vspace{0.2cm}
\noindent\textbf{Performance:} The message length is $L = N^{K_2}$. The total length of answers is
\begin{align}
\sum_{n = 1}^N \ell_n^{[k^*]} = \sum_{\Kc \subset 1:K_2} m(\Kc).
\end{align}
The rate can thus be computed as 
\begin{align}
&R_{\NS}  = \frac{L}{\sum_{n = 1}^N \Eb[\ell_n^{[k^*]}] } \\
& = \frac{N}{N} \frac{N^{K_2-1}}{\sum_{i = 1}^{K_1} \binom{K_1}{i} d_{i, 0} + \sum_{i = 0}^{K_1} \sum_{j = 1}^{K_2 - K_1} \binom{K_1}{i}\binom{K_2 - K_1}{j} d_{i,j} } \\
& = \frac{N^{K_2-1} }{M \frac{N^{K_1} - T_1^{K_1}}{N - T_1} + N^{K_1} \frac{N^{K_2-K_1} - T^{K_2-K_1}}{N - T_2} }\\
& = \frac{N^{K_2-1} }{T_2^{K_2-K_1} \frac{N^{K_1} - T_1^{K_1}}{N - T_1} + \left((T_1 - T_2) \frac{N^{K_1} - T^{K_1}}{N- T_1} + N^{K_1}\right) \frac{N^{K_2 - K_1} - T_2^{K_2 - K_1}}{N - T_2} }\\
&= \left( 1 + \frac{T_1}{N} + \cdots + \left( \frac{T_1}{N} \right)^{K_1 - 1} + \left( \frac{T_1}{N} \right)^{K_1} \left( 1 + \frac{T_2}{N} + \cdots + \left( \frac{T_2}{N} \right)^{K_2 - K_1 - 1} \right) \right)^{-1}.
\end{align}

{\noindent\textit{Remark:} In the general code construction, the message length is $N^{K_2}$. The message length can be further reduced as long as the length of each non-overlapping segments in Group-wise MDS coding step share a maximum common divisor greater than 1. For the example of the NS scheme for $(N, T_1: K_1, T_2:K_2)=(4,2:2,1:4)$ two-level PIR we discussed in Section \ref{sec:eg-NS}, the message length $L = 64 = N^{K_2} / 4$. It is the same for the NB general scheme we will present in the next section and the example of the NB scheme illustrated in Section \ref{sec:eg-NB}.
}

\section{The Non-uniform Block Cancellation scheme}\label{sec:NB}
From the example in Section \ref{sec:eg-NB}, the proposed NB coding scheme uses the $T$-private code discussed in Section \ref{sec:review} as base codes, and consists of three blocks. The NS coding scheme studied in the previous section naturally degrades to the $T$-private code when $K_1=K$ and $T_1=T_2=T$, thus it is leveraged directly in the NB coding scheme. We first construct two precoded tables, which correspond to the NS codes for messages $W_{1:K_1}$ with privacy level $T_1$ and messages $W_{K_1+1:K_2}$ with privacy level $T_2$, respectively. Then a portion of the precoded Table-A is placed in the first block of NB code, a portion of the precoded Table-B is placed in the second block, and the rest of both precoded tables are mixed and form the third block.

\subsection{Precoded tables}

The message length for the NB coding scheme is $L = N^{K_2}$ for the $(N, T_1:K_1, T_2:K_2)$ two-level PIR system. The NS code proposed in Section \ref{sec:NS} for the $(N, T_1:K_1, T_1:K_1)$ two-level PIR consists of $K_1$ layers and has a  message length $N^{K_1}$. Since the message length here is $L = N^{K_2}=N^{K_1}N^{K_2-K_1}$, the NS code can be applied here by stacking the parameters $(m(\cdot), n_1(\cdot), k_1(\cdot), n_2(\cdot), k_2(\cdot))$ by a factor of $N^{K_2-K_1}$. We shall view this coding structure as precoded Table-A. Similarly, define the NS code with message length $N^{K_2}$ for the $(N, T_2:K_2-K_1, T_2:K_2-K_1)$ two-level PIR with messages $W_{K_1+1:K_2}$ as precoded Table-B.

In the precoded Table-A, there are $K_1$ layers of precoded sums. The precoded sums can be indicated by the set of messages involved. Here $\tilde{m}(\Kc_1)$ summations of composition $\Kc_1$ are placed in the $|\Kc_1|$-th layer for any non-empty subset $\Kc_1 \subset 1:K_1$, where
\begin{align}
\tilde{m}_1(\Kc_1) = N^{K_2-K_1+1} (N - T_1)^{|\Kc_1| - 1} T_1^{K_1 - |\Kc_1|}.
\end{align}
Similarly, there are $K_2-K_1$ layers in the precoded Table-B, and $\tilde{m}(\Kc_2)$ summations of compositions $\Kc_2$ placed in the $|\Kc_2|$-th layer for any non-empty subset $\Kc_2 \subset K_1+1:K_2$, where
\begin{align}
\tilde{m}_2(\Kc_2) = N^{K_1+1} (N - T_2)^{|\Kc_2| - 1} T_2^{K_2 - K_1 - |\Kc_2|}.
\end{align}

When the identity of the desired message $k^*$ satisfies $k^* \in K_1+1:K_2$, the precoded Table-B is well-defined, and the precoded Table-A is a \textit{pure-interference table} specified as follows.
\begin{enumerate}
\item \textit{Precoding:} Let $S_{1:K_1}$ be $K_1$ independent random matrices, which are uniformly drawn from the set of all $N^{K_2} \times N^{K_2}$ full rank matrices over $\Fb_q$; these matrices are known only to the user. The precoded messages $W^*_{1:K_1}$ are
\begin{align}
W^*_{k} = S_{k} W_{k}, \quad \forall k \in 1:K_1.
\end{align}
\item \textit{Group-wise MDS coding:} The precoded messages are partitioned into non-overlapping segments, and each segment is MDS-coded under certain appropriate parameters. The coding groups are indexed by non-empty sets $\Kc_1 \subset 1:K_1$. For any non-empty set $\Kc_1 \subset 1:K_1$, for each $k \in \Kc_1$, a segment of $W^*_k$ with length $\frac{T_2}{N} \tilde{m}_{1}(\Kc_1)$ is $(\tilde{m}_1(\Kc_1), \frac{T_2}{N} \tilde{m}_{1}(\Kc_1))$ MDS-coded into $\tilde{W}_k(\Kc_1)$.

\item \textit{Forming pre-coded message sums:} From the $1^{st}$ layer to the $K_1$-th layer, the summations (vectors) placed in the $i$-th layer are formed as
\begin{align}
\sum_{k \in \Kc_1} \tilde{W}_k(\Kc_1), \quad \forall \Kc_1 \subset 1:K_1 \text{ with } |\Kc_1| = i,
\end{align}
for $i = 1, 2, \ldots, K_1$.
\end{enumerate}
We can similarly define the pure-interference precoded Table-B, with the MDS coding parameters $(\tilde{m}_2(\Kc_2), \frac{T_2}{N}\tilde{m}_2(\Kc_2))$ for any coding group indexed by a nonempty set $\Kc_2 \subset K_1+1:K_2$.

\subsection{Encoding, decoding, privacy, and performance}

\vspace{0.2cm}
\noindent\textbf{Encoding:} When the identity of the desired message $k^* \in 1:K_1$, the precoded Table-A is an $N^{K_2-K_1}$-stacked NS code and the precoded Table-B is a pure-interference table. When $k^* \in K_1+1:K_2$, the precoded Table-A is a pure-interference table while the precoded Table-B is an $N^{K_1}$-stacked NS code. The three blocks of NB code are specified as follows.

In precoded Table-A, there are $\tilde{m}_{1}(\Kc_1)$ precoded summations indexed by $\Kc_1$ for any non-empty set $\Kc_1 \subset 1:K_1$. For each non-empty set $\Kc_1 \subset 1:K_1$, $\frac{T_2}{N}$ fractions of the summations indexed by $\Kc_1$ are placed in the $|\Kc_1|$-th layer of the first block. Thus a $\frac{T_2}{N}$ fraction of the precoded Table-A forms the first block. Similarly, a $\frac{T_2}{N}$ fraction of the precoded Table-B forms the second block. The remaining $\frac{N - T_2}{N}$ fractions of both tables are mixed together to form the third block by simple pairwise summations in an arbitrary order; in case they have different numbers of remaining coded symbols, these remaining summands are included directly. The summations of each form are partitioned and distributed to $N$ servers uniformly.

\vspace{0.2cm}
\noindent\textbf{Decoding and correctness:} When $k^* \in 1:K_1$, the precoded Table-B is a pure-interference table. Since the coding group indexed by non-empty set $\Kc_2 \subset K_1+1:K_2$ are $(\tilde{m}_{2}(\Kc_2), \frac{T_2}{N} \tilde{m}_{2}(\Kc_2))$ MDS coded, the $\frac{T_2}{N}$ fraction of the precoded summations placed in the second block can cancel the $\frac{N - T_2}{N}$ fraction of the precoded summations placed in the third block. After canceling all the interference signals involving messages $W_{K_1+1:K_2}$, the NB code becomes precoded Table-A, which can recover the desired message $W_{k^*}$. Similarly, when $k^* \in K_1+1:K_2$, the precoded Table-A is a pure-interference table, and the MDS parameters $(n, k)$ again satisfy $n : k = N : T_2$. Thus the interference signals in the first block can cancel the interference signals in the third block, and the the desired message $W_{k^*}$ can be recovered by the remaining precoded Table-B.

\vspace{0.2cm}
\noindent\textbf{Privacy:} When $k^* \in 1:K_1$, any $T_1$ of $N$ servers collude may be able to infer the desired message is in $W_{1:K_1}$. However, since the pure-interference precoded Table-B is mixed to the precoded Table-A arbitrarily in the third block, and precoded Table-A has privacy level $T_1$ for retrieving any message in $W_{1:K_1}$, i.e., even if any $T_1$ of $N$ servers collude, the identity of the request message $W_{k^*}$ in $W_{1:K_1}$ remains private. It is straightforward to verify that $W_{1:K_2}$ have privacy level $T_2$ since both precoded tables are $T_2$-private.

\vspace{0.2cm}
\noindent\textbf{Performance:} The message length is $L = N^{K_2}$. The size of precoded Table-A is
\begin{align}
t_1 = \sum_{\Kc_1 \subset 1:K_1} \tilde{m}_1(\Kc_1) = N^{K_2-K_1+1} \sum_{i = 1}^{K_1} \binom{K_1}{i} (N - T_1)^{i - 1} T_1^{K_1 - i} = \frac{N^{K_1} - T_1^{K_1}}{N - T_1} N^{K_2-K_1+1};
\end{align}
the size of  precoded Table-B is
\begin{align}
t_2 = \sum_{\Kc_2 \subset K_1+1:K_2} \tilde{m}_2(\Kc_2) & = N^{K_1 +1} \sum_{j = 1}^{K_2-K_1} \binom{K_2-K_1}{j} (N - T_2)^{j - 1} T_2^{K_2 - K_1 - j} \notag \\
&= \frac{N^{K_2-K_1} - T_2^{K_2-K_1}}{N - T_2} N^{K_1+1};
\end{align}
and the size of the third block is
\begin{align}
m = \left(1 - \frac{T_2}{N} \right) \max\left(t_1, t_2\right). 
\end{align}
Since the sizes of the first block and second block are $\frac{T_2}{N}t_1$ and $\frac{T_2}{N} t_2$ separately, the rate is thus
\begin{align}
&R_{\NB} = \frac{L}{\sum_{n = 1}^N \Eb[\ell_n^{[k^*]}] } = \frac{L}{ \frac{T_2}{N} t_1 + \frac{T_2}{N} t_2 + m} = \frac{L}{ t_1 + t_2 + m - \left(1 - \frac{T_2}{N} \right) (t_1 + t_2)}  \notag \\
&=\frac{N^{K-1}}{\max\left(\frac{N^{K_2-K_1} - T_2^{K_2-K_1}}{N - T_2} N^{K_1}+\frac{N^{K_1} - T_1^{K_1}}{N - T_1} N^{K_2-K_1 - 1}T_2, \frac{N^{K_1} - T_1^{K_1}}{N - T_1} N^{K_2-K_1} + T_2\frac{N^{K_2-K_1} - T_2^{K_2-K_1}}{N - T_2} N^{K_1 - 1} \right)} \notag \\
&= \max \left(D^*(K_2-K_1, T_2) + \frac{T_2}{N} D^*(K_1, T_1), D^*(K_1, T_1) + \frac{T_2}{N} D^*(K_2-K_1, T_2) \right)^{-1}.
\end{align}
\vspace{0.2cm}

\section{Conclusion} \label{sec:con}
We considered two-level private information retrieval systems, and provided a capacity lower bound by proposing two novel code constructions and a capacity upper bound. It is further shown that the upper bound can be improved in a special case, however the improved bound also does not match the proposed lower bound. We suspect the proposed code constructions can also be improved to yield better lower bounds, which we leave as a future work. Some of the techniques given in this work can be adopted to multilevel PIR with more than two privacy levels, and when storage constraint is introduced. The two-level model can be viewed as natural generalization of the canonical PIR model. In addition to the extensions and generalizations we discussed in the introduction section, there have been other PIR models in the literature, such as private computation \cite{Sun2018privatecomputation}, PIR with side information \cite{Heidarzadeh:2021role}, and weakly private information retrieval \cite{samy2021asymmetric}. The multilevel privacy model we proposed here can also be further extended to such scenarios. 

\section*{Appendix}
\begin{proof}[Proof of Lemma \ref{lem:property-mnk}]
It is straightforward to verify that
\begin{align}
d_{i, j} : d_{i+1, j} = T_1 : (N - T_1), \quad \forall i \in 0:K_1-1, \forall j \in 0: K_2-K_1, ~\text{with }i+j > 0
\end{align}
and
\begin{align}
d_{i, j}: d_{i, j+1} \geq T_2 : (N - T_2), \quad \forall i=0,\ldots,K_1, \forall j = 0,\ldots, K_2 - K_1 - 1.
\end{align}
\begin{enumerate}
\item For any non-empty set $\Kc \subset 1:K_2$, let $i = |\Kc \cap 1:K_1|$ and $j = |\Kc \cap K_1+1:K_2|$, if $k^* \in 1:K_1$,
\begin{align}
k_1(\Kc) = T_1 d_{i,j} + T_1 d_{i+1, j} = T_1 d_{i,j} + T_1\frac{N - T_1}{T_1} d_{i,j} = N d_{i,j} = m(\Kc);
\end{align}
if $k^* \in K_1+1:K_2$,
\begin{align}
k_2(\Kc) = T_2 d_{i,j} + T_2 d_{i, j+1} \leq T_2 d_{i,j} + T_2\frac{N - T_2}{T_2} d_{i,j} = N d_{i,j} = m(\Kc);
\end{align}
\item If $k^* \in 1:K_1$,
\begin{align}
&\sum_{\Kc \subset 1:K_2,~k^* \in \Kc} m(\Kc) = \sum_{i = 0}^{K_1 -1} \binom{K_1 -1}{i} N d_{i+1, 0} + \sum_{i = 0}^{K_1 -1} \sum_{j = 1}^{K_2-K_1} \binom{K_1 -1}{i} \binom{K_2-K_1}{j} N d_{i+1, j} \\
&= M N^{K_1} + N^{K_1} \frac{N - T_1}{N - T_2} \left(N^{K_2-K_1} - T_2^{K_2-K_1} \right) = N^{K_1} N^{K_2-K_1} = N^{K_2}.
\end{align}
If $k^* \in K_1+1:K_2$,
\begin{align}
\sum_{\Kc \subset 1:K_2,~k^* \in \Kc} m(\Kc) = \sum_{i = 0}^{K_1} \sum_{j = 0}^{K_2-K_1-1} \binom{K_1}{i} \binom{K_2 -K_1 - 1}{j} N d_{i, j+1} = N^{K_2}.
\end{align}
\item When $k^* \in 1:K_1$, if $k \in 1:K_1$, by $k_1(\Kc) = m(\Kc)$,
\begin{align}
&\sum_{\Kc \subset 1:K_2\slash \{k^*\},~ k \in \Kc}  k_1(\Kc) = \sum_{i=1}^{K_1 - 1} \binom{K_1 - 2}{i - 1} Nd_{i, 0} + \sum_{i=1}^{K_1 - 1} \sum_{j = 1}^{K_2-K_1} \binom{K_1 - 2}{i - 1} \binom{K_2-K_1}{j} Nd_{i, j} \\
&= MT_1 N^{K_1 -1} + T_1(N - T_1) \frac{N^{K_1 - 1}}{N - T_2} \left(N^{K_2-K_1} - T_2^{K_2-K_1} \right) = T_1 N^{K_2 -1} < N^{K_2},
\end{align}
and similarly, if $k \in K_1+1:K_2$,
\begin{align}
\sum_{\Kc \subset 1:K_2\slash \{k^*\},~ k \in \Kc} k_1(\Kc) = \sum_{i=0}^{K_1 - 1} \sum_{j = 1}^{K_2-K_1} \binom{K_1 - 1}{i} \binom{K_2 -K_1 - 1}{j - 1} Nd_{i, j} = T_1 N^{K_2 -1} < N^{K_2}.
\end{align}
When $k^* \in K_1+1:K_2$, if $k \in 1:K_1$, 
\begin{align}
&\sum_{\Kc \subset 1:K_2\slash \{k^*\},~k \in \Kc} k_2(\Kc) \notag\\
&= \sum_{i=1}^{K_1} \binom{K_1 - 1}{i - 1} T_2 d_{i, 0} + \sum_{i = 1}^{K_1} \sum_{j = 1}^{K_2 - K_1 - 1} \binom{K_1 - 1}{i - 1} \binom{K_2 -K_1 - 1}{j} T_2 d_{i, j} \notag \\
&\quad\quad\quad\quad\quad\quad\quad \quad\quad + \sum_{i = 1}^{K_1} \sum_{j = 0}^{K_2 -K_1 - 1} \binom{K_1 - 1}{i - 1} \binom{K_2 -K_1 - 1}{j} T_2 d_{i, j+1}\\
&= M T_2 N^{K_1-1} + (N-T_1)T_2N^{K_1-1} \frac{N^{K_2-K_1} - T_2^{K_2-K_1}}{N - T_2} + (N - T_1)T_2N^{K_2 - 2} \\
&=T_2 N^{K_2-1} + (N-T_1)T_2 N^{K_2-2} < N^{K_2},
\end{align}
and similarly, if $k \in K_1+1:K_2$,
\begin{align}
\sum_{\Kc \subset 1:K_2\slash \{k^*\},~k \in \Kc} k_2(\Kc) &= \sum_{i=0}^{K_1} \sum_{j = 1}^{K_2-K_1 - 1} \binom{K_1}{i} \binom{K_2 -K_1 - 2}{j - 1} T_2(d_{i, j}+d_{i,j+1}) \notag \\
& = T_2 N^{K_2 -1} < N^{K_2}.
\end{align}
\end{enumerate}
\end{proof}

\bibliographystyle{IEEEtran}

\end{document}